\ifx\weAreInMain\undefined                                              
\documentclass[12pt]{book}
                                          
\usepackage[utf8]{inputenc}
\usepackage[T1]{fontenc}
\usepackage{amsmath}
\usepackage{amssymb}
\usepackage{epsfig}
\usepackage{float}
\usepackage{subfig}
\usepackage{geometry}
\usepackage{graphicx}
\usepackage{listings}
\usepackage{lmodern}
\usepackage{mdframed}
\usepackage[numbers]{natbib}
\usepackage{siunitx}%
\usepackage{xcolor}
\usepackage{pgfplots}
\usepackage{mathtools}
\usepackage{physics}
\usepackage{setspace}
\usepackage{chngcntr}
\usepackage{enumitem}
\usepackage{titlesec}
\usepackage{lineno}
\usepackage[subfigure]{tocloft}

\usepackage[nottoc,notlot,notlof]{tocbibind} 

\usepackage{afterpage}

\usepackage{imakeidx}
\makeindex

\definecolor{EPCPbooktitlecolor}{RGB}{250,237,204} 
\definecolor{EPCPbooktitlecolorbright}{RGB}{255,249,231} 
\definecolor{darkblue}{RGB}{58,0,184} 
\definecolor{EPCPviolet}{RGB}{71,18,202}
\definecolor{EPCPocher}{RGB}{130,76,1} 

\usepackage[textsize=footnotesize,textwidth=2.8cm]{todonotes}
\ifx\tmpmacrosfolder\undefined
  \newcommand{\tmpmacrosfolder}{.}%
\fi
\ifx\weAreInMain\undefined
{}
\else
\input{\tmpmacrosfolder/tmp_macros_chapters.tex}%
\fi

\usepackage[hidelinks]{hyperref}
\hypersetup{
	colorlinks=true,
	linkcolor=darkblue,
	filecolor=magenta,
    citecolor=darkblue,
	urlcolor=darkblue,
	pdftitle={},
}
\newcommand{\includechapterauthors}[1]{
\par\noindent
\textbf{Authors:} {#1}\\
}

\newcommand{\chapteracknowledgments}[1]{
\par\noindent
\textbf{Contributions:} {#1}\\[-3ex]
}

\newcommand{\chapterkeywords}[1]{
\par\noindent
\textbf{Keywords:} {#1}\\[-3ex]
}

\newcommand{\comment}[1]{{\color{red}#1}}

\newcommand{\chapterroot}{./chapters/TEMPLATE/}

\newcommand{\copyrightinfo}[1]
  {{\color{magenta}\textsc{Copyright Info:} #1}}

\newcommand{\hideinbooklayout}[1]{#1}

\newcommand{\orcid}[1]{}
\newcommand{\affiliations}[1]{}
\newcommand{\affiliationnum}[1]{}

\newtoggle{isdraftversion}

\toggletrue{isdraftversion}

\ifx\isDraft\undefined
\togglefalse{isdraftversion}
\fi

\ifx\weAreInMain\undefined
\toggletrue{isdraftversion}
\fi

\newtoggle{singlechapter}

\togglefalse{singlechapter}

\ifx\weAreInMain\undefined
\toggletrue{singlechapter}
\fi

\ifx\singleChapter\defined
\toggletrue{singlechapter}
\fi

\newtoggle{isprintversion}

\ifx\printVersion\undefined
\togglefalse{isprintversion}
\else 
\toggletrue{isprintversion}
\fi

\newtoggle{ispreprintversion}
\togglefalse{ispreprintversion}

\ifx\preprintVersion\undefined
\togglefalse{ispreprintversion}
\else
\toggletrue{ispreprintversion}
\togglefalse{isprintversion}
\fi

\newtoggle{slidesexist}
\ifx\slidesExist\undefined
\togglefalse{slidesexist}
\else
\toggletrue{slidesexist}
\fi

\newtoggle{showabstract}
\togglefalse{showabstract}

\iftoggle{isdraftversion}
{\toggletrue{showabstract}}
{}

\iftoggle{ispreprintversion}
{\toggletrue{showabstract}}
{}

\iftoggle{showabstract}
{\iftoggle{ispreprintversion}  
	{
		\newcommand{\abstract}[1]{{\centerline{\textbf{Abstract}}: #1}\ \\[2mm]}
	}
	{
		\newcommand{\abstract}[1]{{\noindent\textbf{Abstract}: #1}}
	}
}
{\newcommand{\abstract}[1]{}}

\iftoggle{showabstract}
{\iftoggle{ispreprintversion}  
	{
		\newcommand{\chaptercitation}[3]{\textbf{To cite this chapter:} #1. #2. Version July 2023.
			The Economic Cell Collective. Economic Principles in Cell Biology (2023).
			No commercial publisher | Authors open access book.\ \\[5mm]
			The authors are listed in alphabetical order.
			{}
			\vfill
			\parbox{4cm}{\includegraphics[width=2cm]{../images/qr-code-book-economic-principles-in-biology-black.png}}
			\parbox{12cm}{This is a chapter from the free textbook ``Economic Principles in Cell Biology''.\\
				Free download from \href{https://principlescellphysiology.org/book-economic-principles/index.html}
				{principlescellphysiology.org/book-economic-principles/}.
				\iftoggle{slidesexist}
				{\ \\ Lecture slides for this chapter are available on the website.}}
			
			\parbox{4cm}{
				\href{https://creativecommons.org/licenses/by-sa/4.0/}{\includegraphics[width=3cm]{../images/by-sa.png}}
			}
			\parbox{12cm}{
				\copyright ~2023 The Economic Cell Collective.\\
				Licensed under Creative Commons License CC-BY-SA 4.0.\\
				An authors open access book. No publisher has been paid.        
			}
			\clearpage
			\setcounter{page}{1}}
	}
	{
		\newcommand{\chaptercitation}[3]{\textbf{Citation:} [#1] [#2]}
		\renewcommand{\chapteracknowledgments}[1]{\noindent\textbf{Contributions:} #1}
		\renewcommand{\chapterkeywords}[1]{\noindent\textbf{Keywords:} #1}
	}
}
{
	\newcommand{\chaptercitation}[3]{}
}

\iftoggle{ispreprintversion}
{
	
	\renewcommand{\orcid}[1]{\href{https://orcid.org/#1}{~\includegraphics[width=4mm]{../images/orcid_logo.jpg}}}
	\renewcommand{\affiliations}[1]{\noindent \begin{sloppypar} \setlength{\parindent}{0pt}{#1} \end{sloppypar} \  \\}
	\renewcommand{\affiliationnum}[1]{${}^{#1}$}
}
{
	\iftoggle{isdraftversion}
	{}
	{\iftoggle{ispreprintversion}
		{}
		{\renewcommand{\chapterkeywords}[1]{}
			\renewcommand{\chapteracknowledgments}[1]{}
		}
	}
}

\iftoggle{isdraftversion}
{
	
	\renewcommand{\orcid}[1]{\href{https://orcid.org/#1}{~[ORCID:#1]}}
	\renewcommand{\affiliations}[1]{\begin{sloppypar} \setlength{\parindent}{0pt}{#1} \end{sloppypar} \ \\}
	\renewcommand{\affiliationnum}[1]{${}^{#1}$}
}
{}

\iftoggle{isdraftversion}
{}
{
	\renewcommand{\hideinbooklayout}[1]{}
	\renewcommand{\copyrightinfo}[1]{}
	\renewcommand{\todo}[1]{}
	\renewcommand{\comment}[1]{}
}

\newlistof{allboxes}{lobox}{List of economic analogies, philosphical remarks, experimental methods, mathematical details and other boxes}%
\floatstyle{plain}%
\newfloat{EPCPFloatBlock}{b}{loepcpb}[chapter]

\mdfdefinestyle{EPCPBlocks}{%
  linecolor=\EPCPBlockLineColor,%
  frametitlerule=true,%
  frametitlebackgroundcolor=\EPCPBlockTitleColor,%
  backgroundcolor=\EPCPBlockBGColor,%
  innertopmargin=9pt,
  innerbottommargin=9pt,
  skipabove=7pt,
  skipbelow=7pt
}
\newmdenv[style=EPCPBlocks]{EPCPBlockMD}%

\makeatletter
\newcommand{\EPCPBlockTitleFont}{\normalfont\bfseries}%
\newcommand{\EPCPBlockTitleFontColor}{EPCPocher}
\newcommand{\EPCPBlockLineColor}{EPCPocher}
\newcommand{\EPCPBlockTitleColor}{EPCPbooktitlecolor}
\newcommand{\EPCPBlockBGColor}{EPCPbooktitlecolorbright}

\newmdenv[style=EPCPBlocks,frametitle={\color{\EPCPBlockTitleFontColor}Chapter highlights}]
{chapterhighlights}

\newcommand{\EPCP@FloatSpec}{X}%
\newcommand{\EPCP@BlockTitle}{}%
\newcommand{\EPCP@BlockLabel}{}%
\define@key{EPCPBlocks}{label}{\renewcommand{\EPCP@BlockLabel}{#1}}%
\define@key{EPCPBlocks}{title}{\renewcommand{\EPCP@BlockTitle}{:~{\EPCPBlockTitleFont#1}}}%
\define@key{EPCPBlocks}{float}{\renewcommand{\EPCP@FloatSpec}{#1}}%
\newcounter{dummycounter}

\newenvironment{EPCPBlock}[3][]%
{
  \setkeys{EPCPBlocks}{#1}%
  \ifthenelse{\equal{\EPCP@FloatSpec}{X}}
  {
  \begin{EPCPBlockMD}%
  [frametitle={\color{\EPCPBlockTitleFontColor}#2#3\EPCP@BlockTitle}]
  }
  {
  \floatplacement{EPCPFloatBlock}{\EPCP@FloatSpec}%
  \begin{EPCPFloatBlock}\begin{EPCPBlockMD}%
    [frametitle={\color{\EPCPBlockTitleFontColor}#2#3\EPCP@BlockTitle}]%
  }%
}
{
  \ifthenelse{\equal{\EPCP@FloatSpec}{X}}
  {
  \end{EPCPBlockMD}
  }
  {
  \end{EPCPBlockMD}\end{EPCPFloatBlock}
  }
}

\newcounter{epcpboxcounter}
\counterwithin{epcpboxcounter}{chapter}
\renewcommand{\theepcpboxcounter}{\arabic{chapter}.\Alph{epcpboxcounter}}

{\begin{EPCPBlock}[#1]
  {Economic analogy\mbox{~}}
  {\refstepcounter{epcpboxcounter}\theepcpboxcounter\label{\EPCP@BlockLabel}
  \addcontentsline{lobox}{allboxes}{Economic analogy\mbox{~}\theepcpboxcounter\EPCP@BlockTitle\\\mbox{}}}
  }
{\end{EPCPBlock}}

\newenvironment{Ecoblock*}[1][]%
{\begin{EPCPBlock}[#1]{Economic analogy}{}}
{\end{EPCPBlock}}

{\begin{EPCPBlock}[#1]
  {Philosophical remarks\mbox{~}}
  {\refstepcounter{epcpboxcounter}\theepcpboxcounter\label{\EPCP@BlockLabel}
  \addcontentsline{lobox}{allboxes}{Philosophical remarks\mbox{~}\theepcpboxcounter\EPCP@BlockTitle\\\mbox{}}}
  }
{\end{EPCPBlock}}

\newenvironment{Philblock*}[1][]%
{\begin{EPCPBlock}[#1]{Philosophical remarks}{}}
{\end{EPCPBlock}}

{\begin{EPCPBlock}[#1]
  {Experimental methods\mbox{~}}
  {\refstepcounter{epcpboxcounter}\theepcpboxcounter\label{\EPCP@BlockLabel}
  \addcontentsline{lobox}{allboxes}{Experimental methods\mbox{~}\theepcpboxcounter\EPCP@BlockTitle\\\mbox{}}
  }}
{\end{EPCPBlock}}

\newenvironment{Expblock*}[1][]%
{\begin{EPCPBlock}[#1]{Experimental methods}{}}
{\end{EPCPBlock}}

{\begin{EPCPBlock}[#1]
  {Mathematical details\mbox{~}}
  {\refstepcounter{epcpboxcounter}\theepcpboxcounter\label{\EPCP@BlockLabel}
  \addcontentsline{lobox}{allboxes}{Mathematical details\mbox{~}\theepcpboxcounter\EPCP@BlockTitle\\\mbox{}}
  }}
{\end{EPCPBlock}}

\newenvironment{MathDetailblock*}[1][]%
{\begin{EPCPBlock}[#1]{Mathematical details}{}}
{\end{EPCPBlock}}

{\begin{EPCPBlock}[#1]
		{Box\mbox{~}}
		{\refstepcounter{epcpboxcounter}\theepcpboxcounter\label{\EPCP@BlockLabel}
	     \addcontentsline{lobox}{allboxes}{Box\mbox{~}\theepcpboxcounter\EPCP@BlockTitle\\\mbox{}}}}
	{\end{EPCPBlock}}

\newenvironment{Generalblock*}[1][]%
{\begin{EPCPBlock}[#1]{}{}}
	{\end{EPCPBlock}}

\makeatother

\newcounter{exercisecntr}
\counterwithin{exercisecntr}{chapter}

\titleformat{\section}
{\normalfont\large}
{\Large\sffamily\bfseries\thesection}{0.5em}{\Large\sffamily\bfseries}

\titleformat{\subsection}
{\normalfont\large}
{\large\sffamily\bfseries\thesubsection}{0.5em}{\large\sffamily\bfseries}
                                            
\usepackage[textsize=footnotesize, textwidth=2.8cm]{todonotes}
                                            
\usepackage[final
           ]{showkeys}
\definecolor{refkey}{rgb}{.65,.65,.65}
\definecolor{labelkey}{rgb}{.65,.65,.65}

\setlist{noitemsep}

\hypersetup{
 colorlinks=true,
 linkcolor=red,
 citecolor=green,
 filecolor=red,
 urlcolor=blue,
 } 

\usepackage{titlesec} 
\titleformat{\chapter}{\Huge\sffamily\bfseries}{\chaptername ~ \thechapter}{1em}{}

\usepackage{fancyhdr} 
\usepackage[english]{babel}

\usepackage{amsthm}

\newtheorem{lemm}{Lemma}[section]
\newtheorem{theo}{Theorem}[section]
\newtheorem{prop}{Proposition}[section]

\usepackage{textcomp}
\usepackage{amsfonts,multicol,verbatim, graphicx,epsfig,amssymb,amsmath,cases,bm,multirow,url,dsfont,bbm,tabularx}

\usepackage[ruled, linesnumbered, algochapter]{algorithm2e}

\usepackage{makecell}
\usepackage{array}

\begin{document}

\renewcommand{\chapterroot}{./}                                         
\fi                                                                     




\setcounter{chapter}{9}
\chapter{Intelligent Reflecting Surfaces and Next Generation Wireless Systems}
\label{chapt}

\hideinbooklayout{
\begin{center}
  {\large\textbf{Version: \today}}%
\end{center}
}

\includechapterauthors{Yashuai Cao\affiliationnum{1}, Hetong Wang\affiliationnum{2}, Tiejun Lv\affiliationnum{2}, and Wei Ni\affiliationnum{3}}

\affiliations{
${}^1$ Department of Electronic Engineering, Tsinghua University, Beijing 100084, China, and also with the Beijing National Research Center for Information Science and Technology (BNRist), Beijing 100084, China ({caoys@tsinghua.edu.cn}) \\
${}^2$ School of Information and Communication Engineering, Beijing University of Posts and Telecommunications (BUPT), Beijing 100876, China (\{htwang\_61, lvtiejun\}@bupt.edu.cn)\\
${}^3$ Commonwealth Scientific and Industrial Research Organisation, Sydney, NSW 2122, Australia ({wei.ni@data61.csiro.au})
}

\abstract{
	Intelligent reflecting surface (IRS) is a potential candidate for massive multiple-input multiple-output (MIMO) 2.0 technology due to its low cost, ease of deployment, energy efficiency and extended coverage.
	This chapter investigates the slot-by-slot IRS reflection pattern design and two-timescale reflection pattern design schemes, respectively. 
	For the slot-by-slot reflection optimization, we propose exploiting an IRS to improve the propagation channel rank in mmWave massive MIMO systems without need to increase the transmit power budget. Then, we analyze the impact of the distributed IRS on the channel rank.
	To further reduce the heavy overhead of channel training, channel state information (CSI) estimation, and feedback in time-varying MIMO channels, we present a two-timescale reflection optimization scheme, where the IRS is configured relatively infrequently based on statistical CSI (S-CSI) and the active beamformers and power allocation are updated based on quickly outdated instantaneous CSI (I-CSI) per slot. The achievable average sum-rate (AASR) of the system is maximized without excessive overhead of cascaded channel estimation. A recursive sampling particle swarm optimization (PSO) algorithm is developed to optimize the large-timescale IRS reflection pattern efficiently with reduced samplings of channel samples.
}

\noindent \textbf{Keywords:} Intelligent reflecting surface, distributed IRS deployment, slot-by-slot optimization, two-timescale optimization.



\section{Introduction}\label{sec:1}
Massive multiple-input multiple-output (MIMO) and beamforming are two advanced techniques that can improve the performance and efficiency of wireless networks~\cite{6798744}. However, they also pose significant challenges in terms of hardware and software complexity and cost. Besides, extremely high energy consumption prohibits practical deployment. Thus, scalability and compatibility should be considered as the system needs to scale up to accommodate a large number of users and antennas.


From the perspective of single base station (BS), the average power consumption of fifth-generation (5G) BS can reach approximately 3 \textasciitilde 3.5 times than that of fourth-generation (4G) BS~\cite{fall2023towards}. 
From the network perspective, the network densification deployment and massive MIMO application led to a sharp rise in power consumption of 5G systems. According to~\cite{han2020energy}, 5G systems would consume 12 times more power than 4G networks.
Compared to the 4T4R MIMO in 4G Remote Radio Unit (RRU), the active antenna unit (AAU) in 5G BS adopts 64T64R MIMO. With the increasing channels, the maximum power consumption of a single AAU can reach 1000 \textasciitilde 1400W, and each BS usually consists of 3 AAUs.
Meanwhile, the baseband unit (BBU) in 5G BSs also consumes more power due to the computation increment.
Fig.~\ref{fig:1} compares the maximum power consumption of a single BS between 4G and 5G, where the maximum power consumption of AAU/RRU and BBU components in 5G BSs increases 68\% compared with that in 4G.
As reported in~\cite{shurdi20215g}, the energy cost takes over 19\% \textasciitilde 23\% of radio access network (RAN) operating expenses (OpEx). By 2030, the overall power consumption of 5G BSs would indirectly cause 990,404 tonnes of carbon footprint, which is unacceptable for economic and industrial development and environmental protection.
Particularly for millimeter wave (mmWave) massive MIMO networks, signal blockage would inevitably increase the transmitted energy to guarantee the basic communication quality of service (QoS).

\begin{figure}
  \centering
  \includegraphics[width=4in]{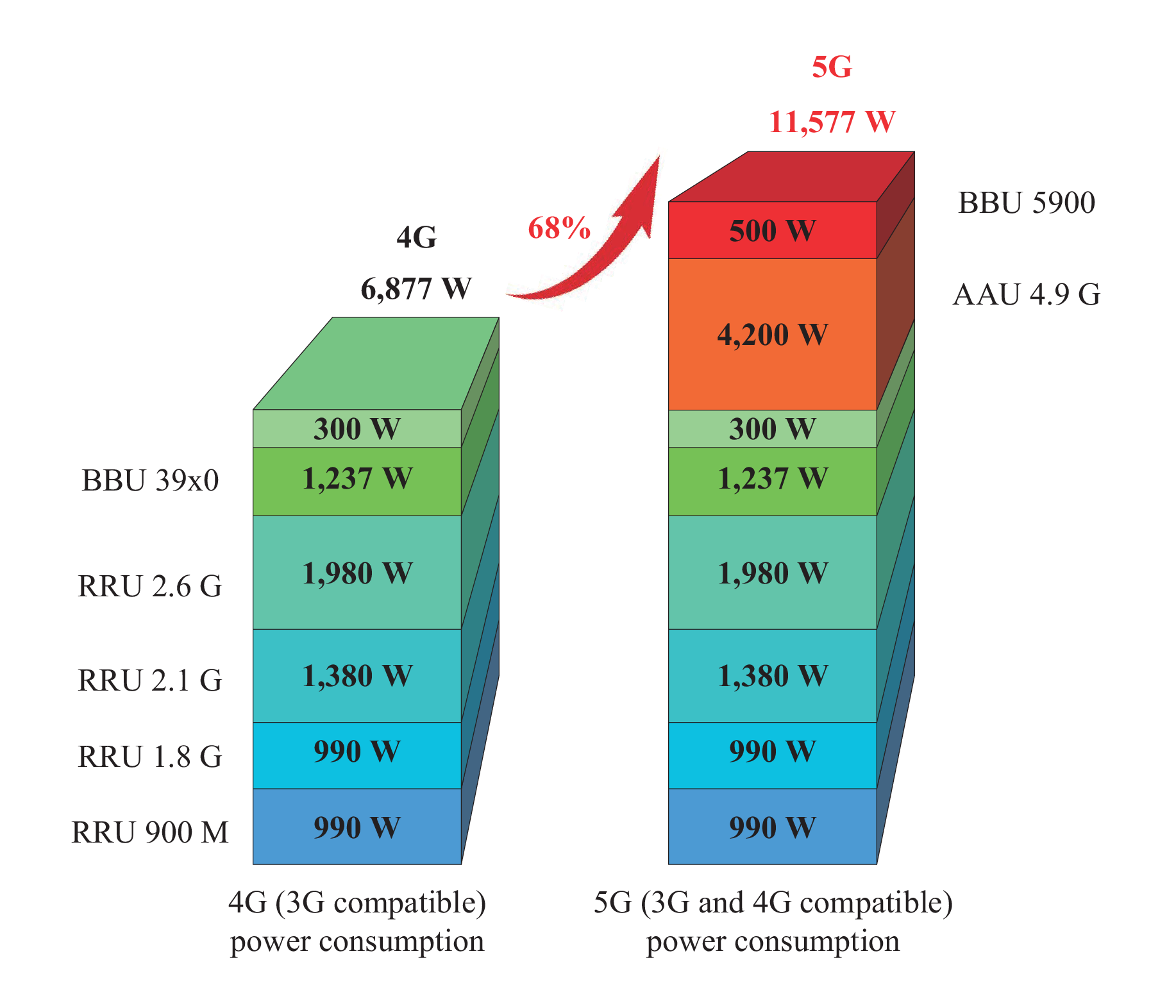}
  \caption{Comparison of power consumption between 4G and 5G BSs~\cite{shurdi20215g}.}
  \label{fig:1}
\end{figure}

Against the background, intelligent reflecting surface (IRS) has been regarded as a promising low-cost technology for massive MIMO systems. 
The concept of IRS has originated from electromagnetic (EM) metamaterial technology based on the generalized Snell's law. IRS comprises massive low-cost passive meta-atoms and a smart controller, which can intentionally adjust the phase shift of reflected EM waves.
The revolutionary idea of IRS is to proactively engineer radio signal propagation while traditional MIMO systems only adapt to wireless environments. Several passive meta-atoms can achieve the ability to control the propagation channel. This means a powerful received signal enhancement with low-cost passive reflective elements instead of active transmit radio frequency (RF) chains.

Table~\ref{Table: 10.1} compares the concept of massive MIMO with IRS. The massive MIMO adopts the operating mechanism of transmitting and receiving signals by RF chains where analog-to-digital converters (ADCs), digital-to-analog converters (DACs), power amplifiers and various hardware are required to realize signal processing. Conversely, the IRS can passively reflect the incident signals and empower the smart radio environment to intelligently reconfigure the wireless propagation environment, which avoids introducing additional noise and leads to negligible power consumption. The optimized IRS reflection modes can achieve constructive or destructive combinations for different users, thus enhancing the overall performance gain of the communication system. Therefore, IRS can be used to promote the journey towards the next level of massive MIMO, i.e., massive MIMO 2.0~\cite{8861014}.


\begin{table}[!ht]
    \centering
    \caption{\centering{Comparison of massive MIMO and IRS}}
    \renewcommand{\arraystretch}{3}
    \begin{tabular}{|l|l|l|}
    \hline
        {\textbf{\makecell[c]{Technology}}} & {\textbf{\makecell[c]{Massive MIMO}}} &{\textbf{\makecell[c]{IRS}}} \\ \hline
        \makecell[c]{Hardware~\cite{8861014}} & \makecell[c]{Several active antennas\\High complexity\\Power-hungry} & \makecell[c]{Massive passive meta-atoms\\Low cost\\Energy-efficiency}\\ \hline
        \makecell[c]{System capacity~\cite{ning2021terahertz}} & \makecell[c]{Logarithmic relationship\\with average transmit power} & \makecell[c]{Linear relationship\\with average transmit power}\\ \hline
         \makecell[c]{Operating\\mechanism~\cite{hu2017potential}} & \makecell[c]{Receive and transmit signal} & \makecell[c]{Reflect signal} \\ \hline
         \makecell[c]{Full-duplex\\communication} & \makecell[c]{RF link required} & \makecell[c]{No RF link required}\\ \hline
         \makecell[c]{Beamforming~\cite{ning2021terahertz}} & \makecell[c]{Fully digital or\\fully analog or\\hybird} & \makecell[c]{Passive and smart}\\ \hline
         \makecell[c]{Power Budget} & \makecell[c]{Very high} & \makecell[c]{Negligible}\\ \hline    
    \end{tabular}
    \label{Table: 10.1}
\end{table}

In support of ubiquitous connectivity, IRS acts in the same role as the active relay in wireless networks.
In particular, IRS is expected to settle the communication coverage issues~\cite{9217160} in high-frequency bands, e.g., mmWave massive MIMO systems.
Therefore, IRS can provide a new paradigm for future 6G ubiquitous connectivity. Naturally, deployments and configuration of IRSs are critical for reaping their full passive beamforming gains.
Existing studies assume the rich-scattering environment between BS and IRS \cite{han2019large}, but the low-rank BS-IRS channel should be considered when it comes to mmWave transmissions.
The sparse multipath channel between the BS and IRS severely limits the degrees of freedom (DoF) of the cascaded channel provided by IRSs. 
The other pressing challenge arising from incorporating IRSs into existing MIMO systems is overcoming the adverse effect of practically imperfect, typically outdated, channel state information (CSI).
Most existing resource allocation algorithms designed for IRS-assisted MIMO systems require the perfect and instantaneous CSI (I-CSI) to perform slot-by-slot IRS optimization, or they can hardly reap the promising gains of the IRSs.
However, slot-by-slot IRS optimization can result in non-negligible overhead and computational complexity in the phase-shift configurations of IRSs~\cite{9806300}.
In general, the phase-shift configurations of the IRSs are closely coupled with MIMO beamformers in IRS-assisted MIMO systems.

In the face of the challenges of IRS deployments and configuration overheads, it is of practical interest to investigate the channel rank improvement and low-overhead IRS reflection configuration schemes in the massive MIMO networks.
The main contributions are listed as follows.
\begin{itemize}
    \item We propose a channel rank improvement method based on a distributed IRS deployment scheme in the mmWave massive MIMO system. By analyzing the DoF of combined channels, we prove the feasibility of a distributed IRS-assisted channel rank improvement scheme. The optimization of each IRS reflection pattern can achieve signal coverage enhancement.
    \item We propose an alternating optimization (AO) framework under the perfect CSI scenario for slot-by-slot sum-rate maximization. In particular, we jointly optimize the active precoding at the BS and the passive reflection pattern of the IRS. Based on the equivalent transformation, we decouple the original optimization problem into two subproblems. For the active precoding problem, we develop a Lagrange multiplier method.  For the reflection optimization problem, we propose a closed-form semi-definite programming (SDP)  method and a low-complexity manifold optimization method to solve the IRS reflection pattern design.
    \item We establish a two-timescale reflection design framework for the imperfect and outdated CSI (O-CSI) scenario, where slot-by-slot optimization overhead is significantly reduced. For the small timescale, we exploit the inherent temporal correlation of the channel between consecutive time slots to suppress the performance degradation incurred by O-CSI. For the large timescale, a modified particle swarm optimization (PSO) method, namely, recursive sampling PSO (rsPSO), is proposed to solve the IRS reflection optimization with reduced complexity.
\end{itemize}
	
The remainder of this chapter is outlined as follows. Section \ref{sec:2} discusses the related studies of IRS-assisted MIMO systems. Section \ref{sec:3} describes the hardware architecture and work principles of IRS. In Section \ref{sec:4}, we present the channel model and distributed-IRS system model, and analyze the channel rank improvement.
In Section \ref{sec:5}, we formulate a sum-rate maximization problem and propose an AO framework to configure the IRS reflection pattern slot-by-slot.
Section \ref{sec:6} proposes a two-timescale IRS reflection pattern design strategy, where the rsPSO-based IRS configuration algorithm at a large timescale is developed. Section \ref{sec:7} provides simulation results and corresponding discussions. Sections \ref{sec:8} and \ref{sec:9} draw conclusions and identify some future works.

\section{Related Work}\label{sec:2}
Recent advancement in IRS shows the vast perspectives across signal modulation~\cite{8928065}, relay transmission~\cite{8811733, 9519722}, target sensing~\cite{9427098}, over-the-air (OTA) analog computation~\cite{9985456}, as illustrated in Fig.~\ref{fig:2}. These new functions depend strongly on the proper IRS reflection pattern optimization.

Existing IRS studies have concentrated on reflection optimization to improve spectral efficiency (SE) and energy efficiency (EE) separately in most scenarios, as summarized in Table~\ref{Table: 10.2}.
The authors of~\cite{han2019large} investigate an IRS-assisted downlink multi-input single-output (MISO) communication system, which shows that the upper bound of the ergodic SE is only proportional to the numbers of IRS atoms and  BS antennas but independent of to the phase-shift values in the case of Rayleigh fading. In~\cite{8811733, huang2019reconfigurable}, the IRS-assisted multi-user MISO network is considered to guarantee the individual QoS of each mobile user (MU). The optimized IRS reflection patterns can attain the same sum-rate performance while improving the EE by up to three-fold compared to relay-aided systems. Most IRS-assisted systems assume that the propagation channel between BS and IRS is rich-scattering~\cite{han2019large, huang2019reconfigurable}. However, for the mmWave massive MIMO scenario, sparsity and low rank are the key features of channels, thus leading to a low-rank cascaded channel matrix and small DoF. Considering the demand for multi-stream transmissions in mmWave communication networks, additional reflection paths are introduced by deploying IRSs to maximize the received signal power~\cite{wang2020intelligent}, minimize system energy consumption~\cite{9707727}, resist jamming attacks~\cite{10021680}, or strengthen the integrity of OTA federated learning~\cite{9829190}. 
Moreover, IRSs have also been integrated into unmanned aerial vehicles (UAVs) to provide uninterrupted air-to-ground connections~\cite{9801642,10002850,10114467}, provide radio surveillance~\cite{10187161}, or combat adversarial jamming attacks~\cite{10146001}.

High-accuracy channel estimation is vital to reaping the promising gains of the IRS.
However, it is rather challenging to estimate the IRS-involved channels because the passive IRS meta-atoms do not involve active RF chains. Although IRS-related cascaded channel estimations have been widely studied in \cite{9087848, 9130088, 9133156}, CSI is always imperfect due to channel estimation errors. In~\cite{zhou2020robust}, the authors focused on robust transmissions under imperfect CSI to satisfy each user's QoS. The above-mentioned works are optimized for the reflection pattern slot by slot. 
To reduce the channel training overhead, a two-timescale transmission protocol is proposed in ~\cite{zhao2020intelligent}, where the sum-rate can be maximized by optimizing IRS and precoder with statistical CSI (S-CSI) and I-CSI. For the larger Rician factor, an S-CSI-based scheme becomes more effective because of more deterministic channels. Nevertheless, the I-CSI is generally dramatically and quickly outdated in time-variant channels. Existing studies generally overlook the temporal correlation between channel slots for IRS design.

\begin{figure}[t!]
    \centering
    \begin{minipage}[t]{1.0\linewidth}
    \centering
        \begin{tabular}{@{\extracolsep{\fill}}c@{}c@{}@{\extracolsep{\fill}}}
            \includegraphics[width=0.4\linewidth]{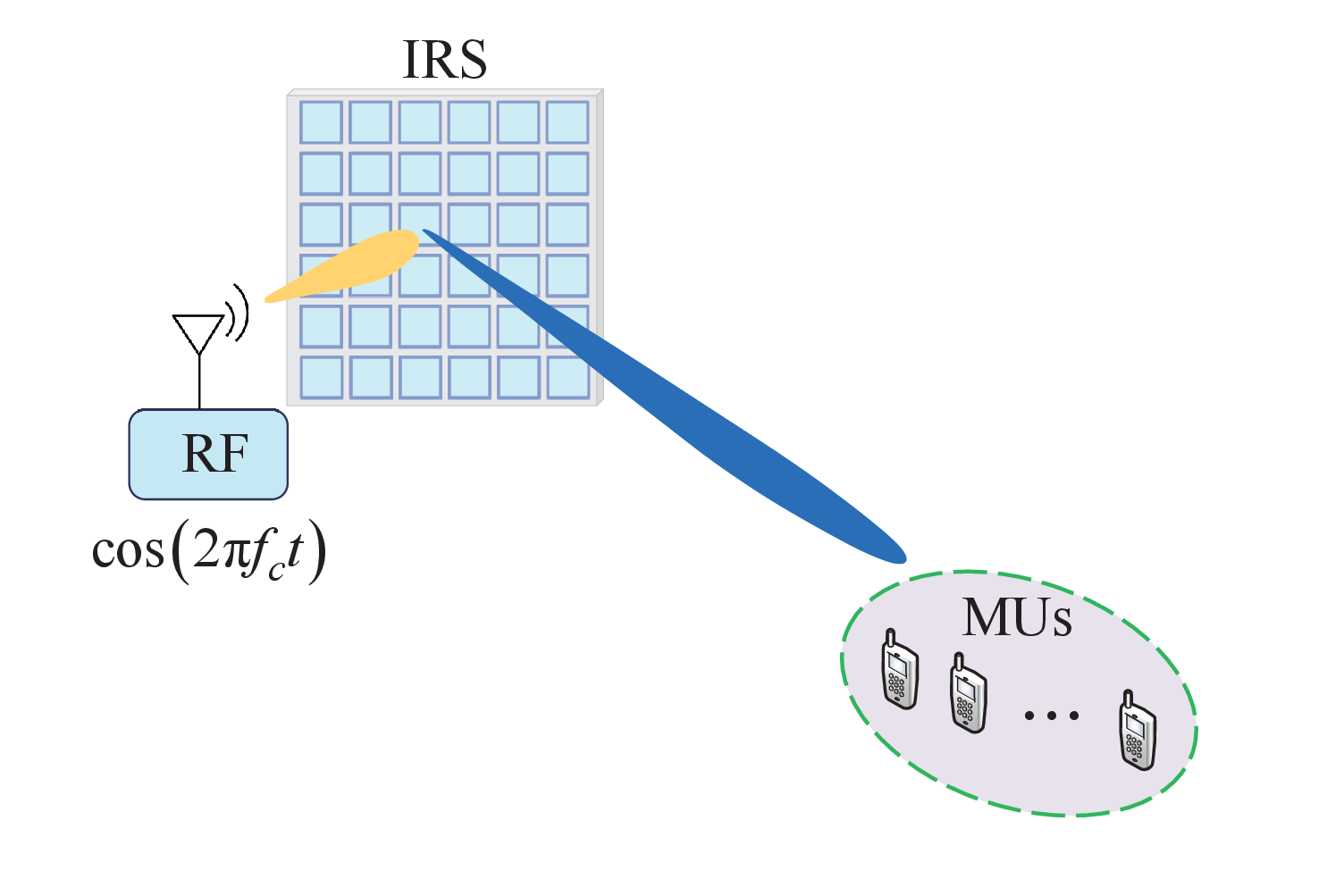} &
            \includegraphics[width=0.4\linewidth]{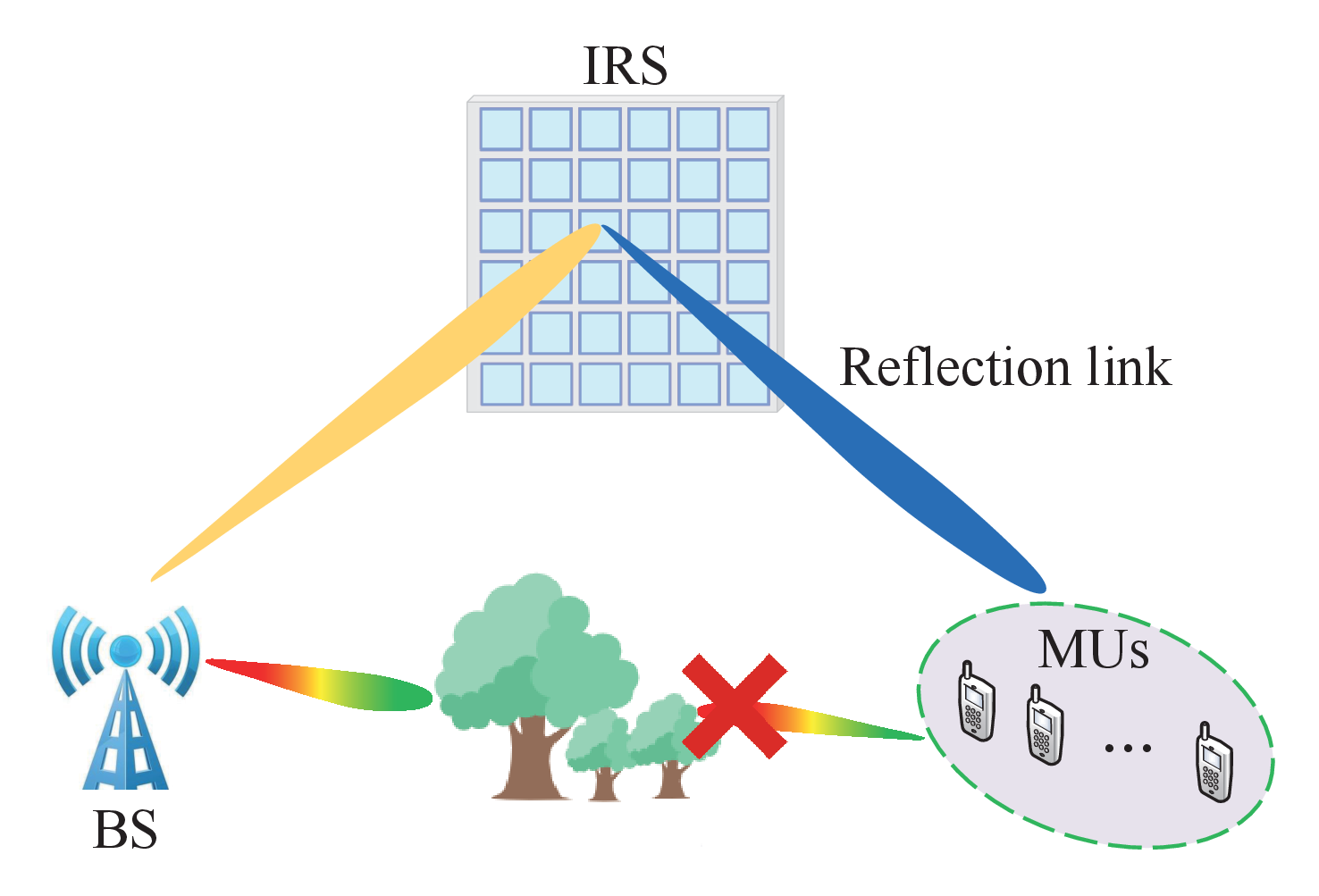}\\
            (a) Signal modulation. & (b) Relay transmission.\\
        \end{tabular}
    \end{minipage}
    \begin{minipage}[t]{1.0\linewidth}
    \centering
        \begin{tabular}{@{\extracolsep{\fill}}c@{}c@{}@{\extracolsep{\fill}}}
            \includegraphics[width=0.4\linewidth]{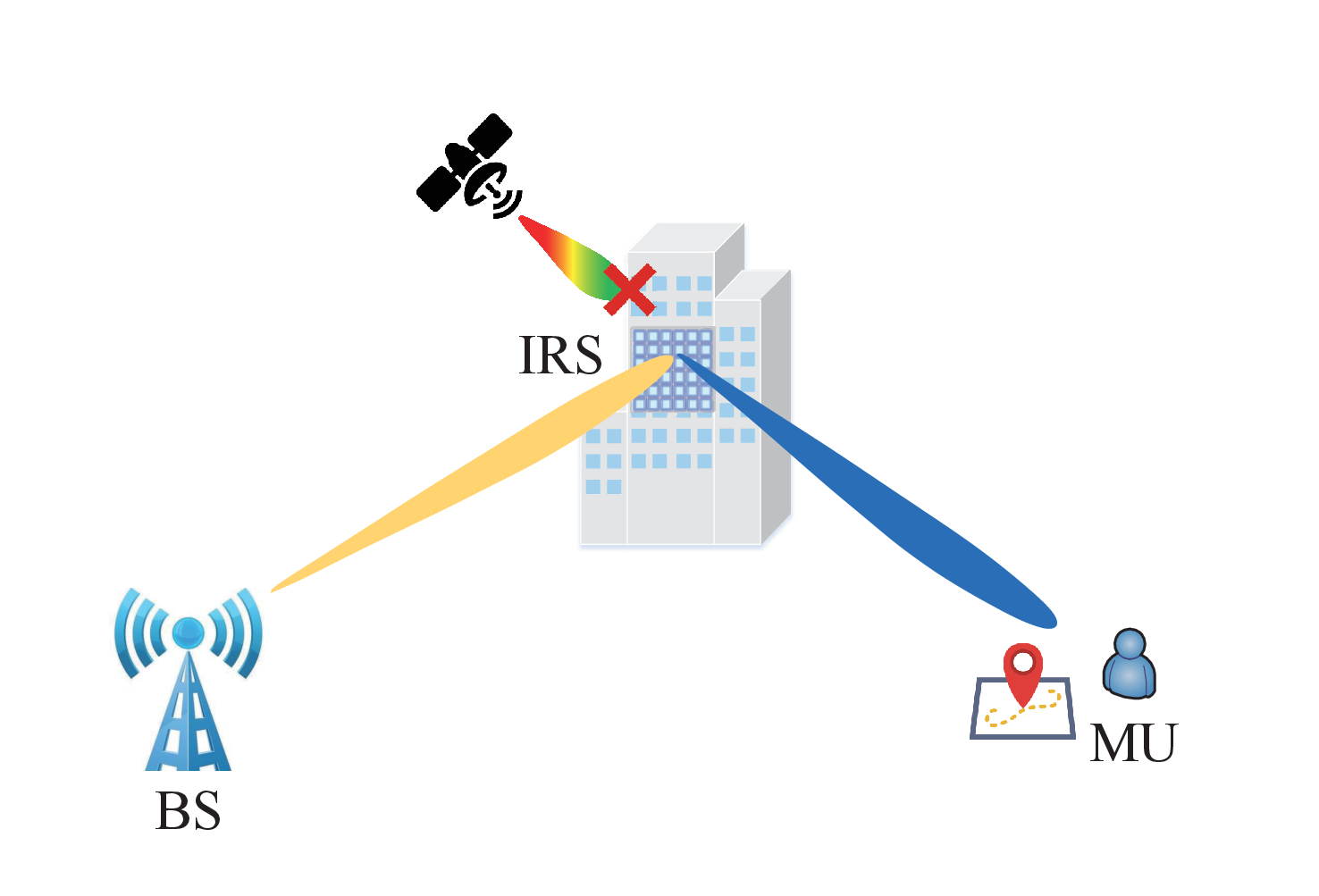} &
            \includegraphics[width=0.4\linewidth]{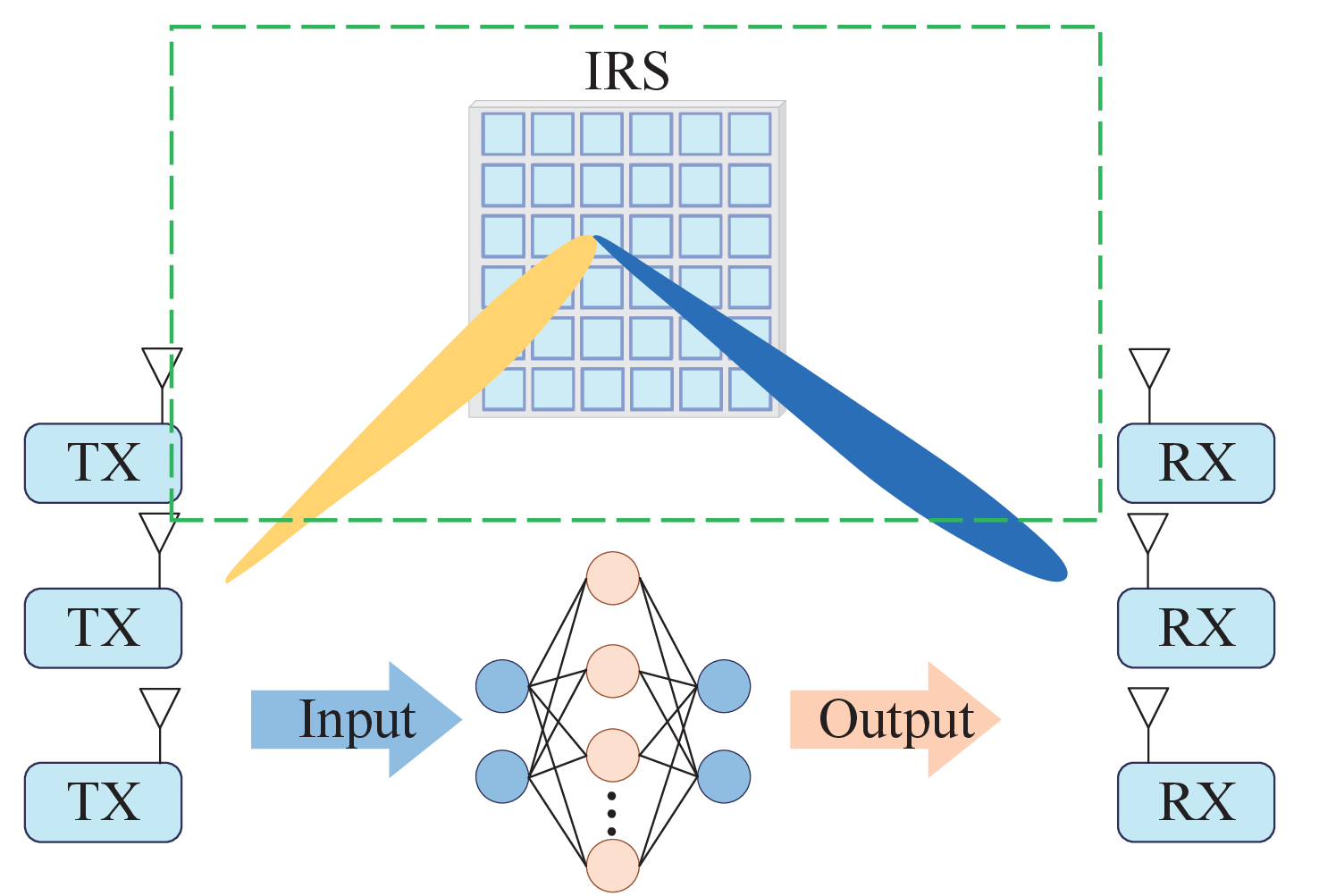}\\
            (c) Target sensing. & (d) OTA analog computation.\\
        \end{tabular}
    \end{minipage}
    \caption{IRS-enabled new functions: (a) modulate the pure carrier and transmit information symbols by IRS reflection; (b) overcome the signal blindspot to ensure robust transmissions; (c) IRS performs beamforming and provides desirable RF propagation properties for sensing; (d) realize the specified computation tasks by controlling the wave propagation.}
    \label{fig:2}
 \end{figure}


\begin{table}[!ht]
    \centering
    \caption{\centering{Literature review of related works}}
    \scalebox{1}{
    \begin{tabular}{|l|l|l|l|l|l|}
    \hline
        \textbf{Ref.} & \textbf{\makecell[c]{Scenario}} & \textbf{\makecell[c]{Optimization\\objectives}} & \textbf{\makecell[c]{Reflection\\pattern\\configuration\\scheme}} & \textbf{\makecell[c]{Role of IRS}} & \textbf{\makecell[c]{CSI}} \\ \hline
        \makecell[c]{~\cite{han2019large}} & \makecell[c]{Single-user \\ MISO} & \makecell[c]{Maximize\\ergodic SE} & \makecell[c]{Slot-by-slot} & \makecell[c]{Coverage \\ enhancement} & \makecell[c]{Perfect\\I-CSI}\\ \hline
        \makecell[c]{~\cite{8811733}} & \makecell[c]{Multi-user\\MISO} & \makecell[c]{Maximize\\transmit power} & \makecell[c]{Slot-by-slot}& \makecell[c]{Hardware cost \\ reduction} & \makecell[c]{Perfect\\I-CSI} \\ \hline
        \makecell[c]{~\cite{huang2019reconfigurable}} & \makecell[c]{Multi-user\\MISO}  & \makecell[c]{Maximize\\EE} & \makecell[c]{Slot-by-slot} & \makecell[c]{Coverage \\ enhancement} & \makecell[c]{Perfect\\I-CSI} \\ \hline
        \makecell[c]{~\cite{wang2020intelligent}} & \makecell[c]{mmWave \\MISO} & \makecell[c]{Maximize\\received power} & \makecell[c]{Slot-by-slot} & \makecell[c]{Channel rank\\improvement} & \makecell[c]{Perfect\\I-CSI} \\ \hline
         \makecell[c]{~\cite{zhou2020robust}} & \makecell[c]{Multi-user\\MISO} & \makecell[c]{Minimize\\transmit power} & \makecell[c]{Slot-by-slot} & \makecell[c]{Save power \\ consumption} & \makecell[c]{Imperfect\\CSI} \\ \hline
          \makecell[c]{~\cite{zhao2020intelligent}} & \makecell[c]{Multi-user\\MISO} & \makecell[c]{Maximize\\expected SE} & \makecell[c]{Two-timescale} & \makecell[c]{Improve\\channel quality} & \makecell[c]{Perfect\\S-CSI\\and I-CSI} \\ \hline
    \end{tabular}}
    \label{Table: 10.2}
\end{table}

\section{Intelligent Reflecting Surfaces Hardware and Architectures}\label{sec:3}
In this section, we first briefly describe the reflection pattern model of IRS and provide the corresponding work principle. Then, the reconfigurable impedance network is modeled to explain how IRS achieves spatially varying phase response.

To elaborately describe the array response model of passive IRS, let us first consider the two-dimensional (2D) array architecture in the three-dimensional (3D) Cartesian coordinate system as Fig. \ref{fig:3}. 
Consider a rectangular IRS of $N$ meta-atoms placed on the $xoy$-plane, with $N_x$ meta-atoms arranged along the $x$-axis and $N_y$ meta-atoms arranged along the $y$-axis. Taking the origin $O$ as the reference center, we define the vector pointing from $O$ to the ($x, y$)-th meta-atom as $\mathbf{b}_{x,y}$ with $x \in \{1,2,\cdots, N_x\}$ and $y \in \{1,2,\cdots, N_y\}$. The unit vector of the $l$-th incident path is denoted by $\mathbf{e}_l$. Then, the phase of the signal incident upon the ($x, y$)-th meta-atom is given by
\begin{align}
\varphi_{x,y}^{\mathrm{i}} = \frac{2\pi \mathbf{b}_{x,y}^{\mathsf{T}} \mathbf{e}_l}{\lambda},
\label{eq:1}
\end{align}
where $\lambda$ is the carrier wavelength. Likely, if $\mathbf{e}_l$ denotes the $l$-th departure path, the phase of the signal departure from the ($x, y$)-th meta-atom is
\begin{align}
\varphi_{x,y}^{\mathrm{d}} = \frac{2\pi \mathbf{b}_{x,y}^{\mathsf{T}} \mathbf{e}_l}{\lambda}.
\label{eq:2}
\end{align}

\begin{figure}[t]
	\centering{}\includegraphics[width=3.5in]{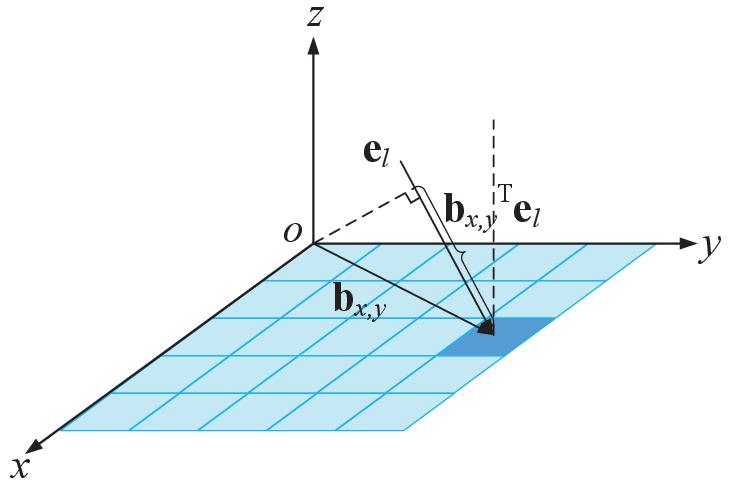}
	\caption{Array geometry structure of IRS.}
	\label{fig:3}
\end{figure}

\par Accordingly, the array response vector of the $l$-th incident/departure path can be written as
\begin{align}
\mathbf{a}_{\mathrm{i}/\mathrm{d}}(\mathbf{e}_l)  =  \exp &\Big( -\mathrm{j} 
\big[ 
\varphi_{1,1}^{\mathrm{i}/\mathrm{d}}, \varphi_{2,1}^{\mathrm{i}/\mathrm{d}}, \cdots, \varphi_{N_x,1}^{\mathrm{i}/\mathrm{d}}, 
\varphi_{1,2}^{\mathrm{i}/\mathrm{d}}, \varphi_{2,2}^{\mathrm{i}/\mathrm{d}}, \cdots, \varphi_{N_x,2}^{\mathrm{i}/\mathrm{d}}, \nonumber\\ 
&\cdots, 
\varphi_{1,N_y}^{\mathrm{i}/\mathrm{d}},\varphi_{2,N_y}^{\mathrm{i}/\mathrm{d}}, \cdots, \varphi_{N_x,N_y}^{\mathrm{i}/\mathrm{d}}
\big]^{\mathsf T} \Big).
\label{eq:3}
\end{align}

\begin{figure}[h]
    \centering{}\includegraphics[scale=0.75]{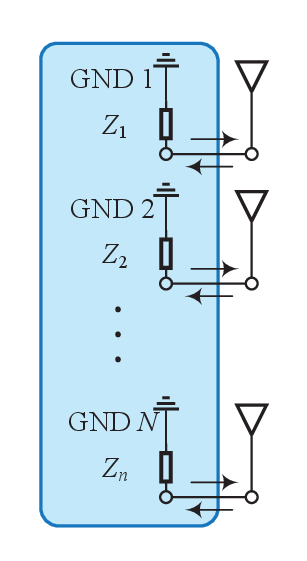}
	\caption{$N$-port single-connected reconfigurable impedance network, which explains the physical significance of the IRS reflection pattern.}
	\label{fig:4}
\end{figure}

IRS can provide the functionality of spatially varying phase response. To explain this physical mechanism conveniently, we model the IRS by a single-connected reconfigurable impedance network, as shown in Fig. \ref{fig:4}. There are $N$ meta-atom connected to the $N$-port reconfigurable impedance network.  Since each port is single-connected without connecting to the other ports, each port connects to the ground through the independently reconfigurable and passive impedance. Consequently, each port can generate an independent phase shift, thus leading to the diagonal scattering matrix $\mathbf{\Phi}$~\cite{9514409}:

\begin{align}
\left[ \mathbf{\Phi} \right]_{n,n} = \frac{Z_{n} - Z_0}{Z_{n} + Z_0}=\frac{\mathrm{j} X_{n} - Z_0}{\mathrm{j} X_{n} + Z_0} = \mathrm{e}^{\mathrm{j} \theta_n}, \quad  \forall n=1,2,\cdots, N,
\label{eq:4}
\end{align}
where $\theta_n$ is the phase response of the $n$-th meta-atom; $\mathbf{\Phi}=\mathrm{diag} \left( [\mathrm{e}^{\theta_1}, \mathrm{e}^{\theta_2}, \cdots, \mathrm{e}^{\theta_N}] \right)$ is a diagonal phase-shift matrix (reflection pattern) of IRS; $Z_n$ and $X_n$ are the reconfigurable impedance and reactance of the $n$-th meta-atom, respectively; and $Z_0$ is the reference impedance.

\begin{theo} \label{theo:1.3.1}
Given the angle $\varphi_{n}^{\mathrm{i}/\mathrm{d}}$ which is the $n$-th entry of vector $\mathbf{a}_{\mathrm{i}/\mathrm{d}}$, the array gain of the IRS from the incident wave to the departure wave is equivalent to:
\begin{align}
G_a(\mathbf{\Phi}) &= \left\vert\mathbf{a}_{\mathrm{i}}^{\mathsf H}\mathbf{\Phi}\mathbf{a}_{\mathrm{d}} \right\vert^2 = \left\vert \sum_{n=1}^{N} \mathrm{e}^{ \mathrm{j} \left( \varphi_{n}^{\mathrm{i}} + \theta_n - \varphi_{n}^{\mathrm{d} }  \right) } \right\vert^2.
\label{eq:5}
\end{align}
The maximum array gain can be achieved by setting
\begin{align}
\theta_n = \varphi_{n}^{\mathrm{d}} - \varphi_{n}^{\mathrm{i}}.
\label{eq:6}
\end{align}
\end{theo}

\begin{proof}
The array gain in \eqref{eq:5} can be reframed as
\begin{align}
G_a(\mathbf{\Phi}) &= \left\vert \sum_{n=1}^{N} \mathrm{e}^{ \mathrm{j} \left( \varphi_{n}^{\mathrm{i}}  - \varphi_{n}^{\mathrm{d} }  \right) } \mathrm{e}^{ \mathrm{j}\theta_n} \right\vert^2 = \left\vert \boldsymbol{\varphi}^{\mathsf{T}} \boldsymbol{\theta} \right\vert^2,
\label{eq:7}
\end{align}
where $\boldsymbol{\varphi} = [\mathrm{e}^{ \mathrm{j} \left( \varphi_{1}^{\mathrm{i}}  - \varphi_{1}^{\mathrm{d} }  \right)}, \mathrm{e}^{ \mathrm{j} \left( \varphi_{2}^{\mathrm{i}}  - \varphi_{2}^{\mathrm{d} }  \right)}, \cdots, \mathrm{e}^{ \mathrm{j} \left( \varphi_{N}^{\mathrm{i}}  - \varphi_{N}^{\mathrm{d} }  \right)} ]^{\mathsf T}$ and $\boldsymbol{\theta}=[\mathrm{e}^{ \mathrm{j}\theta_1}, \mathrm{e}^{ \mathrm{j}\theta_2}, \cdots, \mathrm{e}^{ \mathrm{j}\theta_N}]^{\mathsf T}$. Defining $\varphi'_n = \mathrm{e}^{ \mathrm{j} \left( \varphi_{n}^{\mathrm{i}}  - \varphi_{n}^{\mathrm{d} }  \right)}$, we note that
\begin{align}
\left\vert \boldsymbol{\varphi}^{\mathsf{T}} \boldsymbol{\theta} \right\vert^2 \le \left\vert \sum_{n=1}^{N} \vert \varphi'_n \vert  \vert \theta_n \vert\right\vert^2,
\label{eq:8}
\end{align}
which is equal to satisfying the condition of $\theta_n = \arg (\varphi'_n) = \varphi_{n}^{\mathrm{d}} - \varphi_{n}^{\mathrm{i}}$.
\end{proof}

\section{Intelligent Reflecting Surfaces and the Path Loss Model}\label{sec:4}

In this section, we first provide the product-distance path loss model for IRS cascaded channels. Then, we build the communication models for the centralized and distributed IRS deployments in the mmWave massive MIMO system, respectively. Finally, we analyze the channel rank improvement by the distributed deployment.\par
IRS is regarded as an emerging and appealing 6G-enabled technology, according to its passive and intelligent reflection ability. The terminology ``passive'' indicates that meta-atoms only reflect the incident waves immediately without any analogy and digital signal processing and time delay. As such, IRS can be regarded as intelligent scatters in the physical space to improve the propagation. Particular for future high-frequency communication systems with severe propagation attenuation and poor diffraction bottlenecks, IRS is significantly compatible with the high-frequency bands (e.g., mmWave and Terahertz) by providing uninterrupted wireless connectivity with low power consumption~\cite{9373634, 10287315}.

Although IRS can bring promising reflection gain in a low-cost manner, the actual cascaded channel fading should be considered. As shown in Fig.~\ref{fig:5}, the distances from the BS to IRS and from the IRS to MU can be denoted as $d_1$ and $d_2$, respectively. The path loss of BS-IRS and IRS-MU links through the $n$-th meta-atom of the IRS follow $\mathbb{E}\left(|g_n|^{2}\right)\propto c_1\frac{d_1}{d_0}^{-\alpha_1}$ and $\mathbb{E}\left(|h_n|^{2}\right)\propto c_2\frac{d_2}{d_0}^{-\alpha_2}$, where $c_1\left(c_2\right)$ and $\alpha_1\left(\alpha_2\right)$ represent the path loss value at reference distance $d_0$ and the path loss exponent, respectively. As a result, the path loss of the IRS cascaded link can be denoted as $P_r\propto\frac{1}{{d_1}^{a_1}{d_2}^{a_2}}$, which is inversely proportional to product-distance instead of the sum-distance. Consequently, the product-distance path loss effect would lead to double fading attenuation~\cite{10243567}, which should be properly compensated by the IRS reflection gains.

\begin{figure}[t]
	\centering{}\includegraphics[width=3.5in]{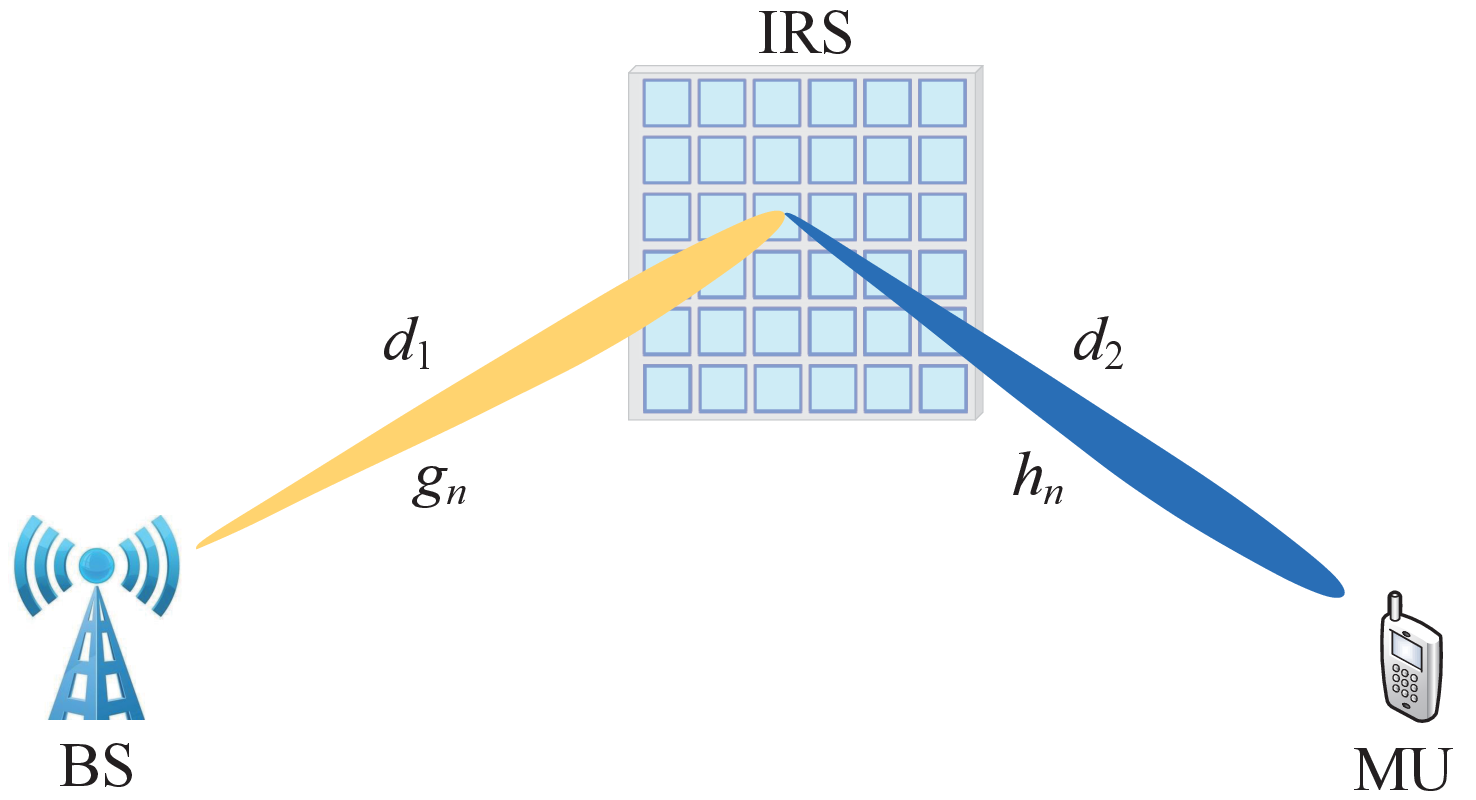}
	\caption{Product-distance path loss model in IRS-assisted mmWave system.}
	\label{fig:5}
\end{figure}
Consider an IRS-assisted downlink mmWave system, as shown in Fig.~\ref{fig:6}. The BS is equipped with a uniform linear array (ULA) composed of $N_{\mathrm{t}}$ antenna elements. With the assistance of the IRS, $K$ single-antenna MUs are served. 
Two deployment schemes of the IRS are provided: centralized and distributed.
In a typical technical scenario of the IRS, the direct mmWave links between BS and MUs are severely blocked by obstacles, while the MUs are located in the hotspot area served by the IRS. Thus, the BS can only communicate with the MUs via the complementary links offered by IRS since the mmWave links between BS and MUs are highly susceptible to environmental blockages and dynamics.

For the centralized deployment scheme, we assume a single IRS is placed near the hotspot. Based on the discussion in Section \ref{sec:3}, we can directly obtain the communication model as:
\begin{align}
\mathbf{y} =& \mathbf{H}^{\mathsf H} \mathbf{\Phi}^{\mathsf H} \mathbf{G} \sum_{k=1}^{K}\mathbf{w}_k {s}_k + \mathbf{n},
\label{eq:9}
\end{align}
where $\mathbf{H}=[\mathbf{h}_1, \mathbf{h}_2, \cdots, \mathbf{h}_K]$ collects the channels of all MUs, $\mathbf{G}$ is the channel between the BS and IRS, $\mathbf{W}=[\mathbf{w}_1, \mathbf{w}_2, \cdots, \mathbf{w}_K]$ is the precoding matrix at the BS, $\mathbf{s}=[{s}_1, {s}_2, \cdots, {s}_K]^{\mathsf{T}}$ collects the intended signals of MUs, and $\mathbf{n}$ is the complex Gaussian noise of $\mathcal{CN}(0,\sigma^2\mathbf{I})$.

\begin{figure}[t]
	\centering{}\includegraphics[width=3.5in]{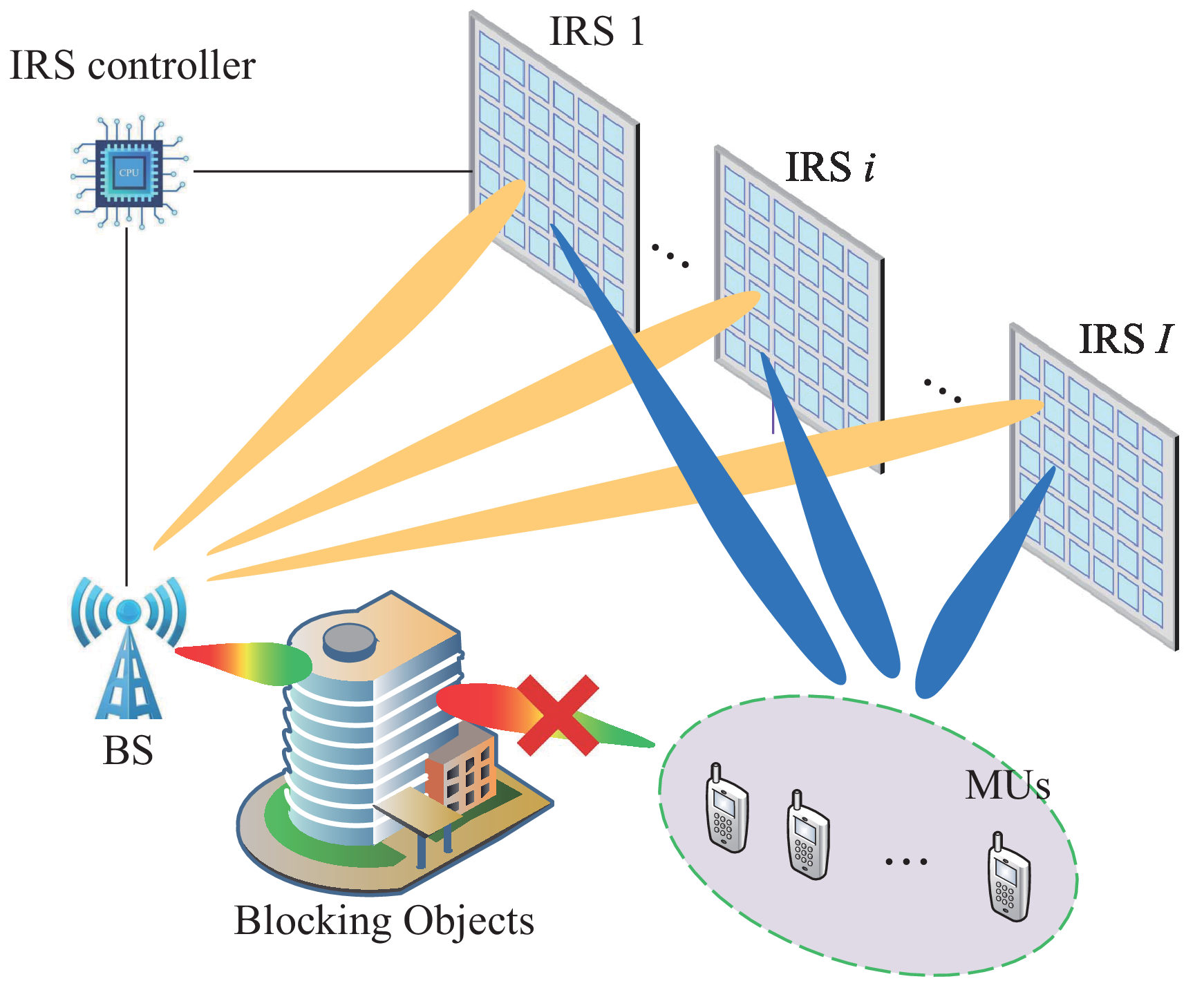}
	\caption{IRS-assisted downlink mmWave massive MIMO system.}
	\label{fig:6}
\end{figure}

According to the widely used Saleh-Valenzuela channel model \cite{6834753} for mmWave communications, $\mathbf G$ can be modeled as
\begin{align}
\mathbf{G} = \sum_{l=0}^{L} \nu^{(l)} \mathbf{a}_{\mathrm{B}} (\mathbf{e}_l)\mathbf{a}_{\mathrm{I}}^{\mathsf{H}} (\mathbf{e}_l), 
\label{eq:10}
\end{align}
where $L$ denotes the number of non-line-of-sight (NLoS) paths, $l=0$ represents the line-of-sight (LoS) path, $\nu^{(l)}$ is the complex gain of the $l$-th path, and $\mathbf{a}_{\rm{B}}$ and $\mathbf{a}_{\rm{I}}$ are the array steering vectors of the BS and IRS.

Typically, the IRS is densely distributed in the hotspot spaces, which gives rise to a high probability of LoS propagation. Due to the severe path loss, the transmit power of 2 or more reflections can be ignored so that only LoS is considered \cite{8811733}.
Similarly, the channel between the IRS and the $k$-th MU is given by
\begin{equation}
\mathbf{h}_{k}= \nu_{k} \mathbf a_{\mathrm{U}}\left( \mathbf{e}_k \right), \forall k=1,2,\cdots, K,
\label{eq:11}
\end{equation}
where $\nu_k$ is the channel gain, and $\mathbf{a}_{\mathrm{U}}$ is the array steering vector of IRS.

However, mmWave channels generally consist of only a small number of dominant multipath components and typically exhibit 3\textasciitilde 5 paths in realistic environments \cite{6847111}, while the scattering at sub-6 GHz is relatively rich.
Therefore, the virtual LoS channel created by the IRS can hardly enjoy the spatially multiplexed benefits. Note that the maximum number of signal streams carried by the virtual LoS channel is determined by the rank of the cascaded channel matrix $\mathbf{\mathbf{H}}^{\mathsf H} \mathbf{\Phi}^{\mathsf H} \mathbf{G}$. When $\mathrm{rank}(\mathbf{G}) < K$, the multi-user communication system does not work.

For this reason, we propose a distributed deployment scheme to achieve channel rank improvement. Specifically, we adopt $I$ IRS units to serve $K$ MUs, and each IRS unit is composed of $N$ meta-atoms. These distributed IRS units are managed by a smart controller, which exchanges CSI and coordinates the reflection pattern for all IRS units.

Then, the received signal at the $k$-th MU can be given by
\begin{align}
\mathbf{y} =& \sum_{i=1}^{I}\mathbf{H}_{i}^{\mathsf H} \mathbf{\Phi}_i^{\mathsf H} \mathbf G_{i} \sum_{k=1}^{K} \mathbf{w}_k {s}_k  + \mathbf{n},
\label{eq:12}
\end{align}
where $\mathbf G_{i}$ is the channel between BS and the $i$-th IRS unit, $\mathbf{H}_{i} =[\mathbf{h}_{i,1}, \mathbf{h}_{i,2}, \cdots, \mathbf{h}_{i,K}]$ is the channel between the $i$-th IRS unit and the $k$-th MU, and the reflection pattern matrix of the $i$-th IRS unit is denoted by $\mathbf{\Phi}_i= \mathrm{diag}( \boldsymbol{\theta}_{i} )$ with $\boldsymbol{\theta}_{i}=[\mathrm{e}^{\theta_{i,1}}, \mathrm{e}^{\theta_{i,2}}, \cdots, \mathrm{e}^{\theta_{i,N}}]^{\mathsf T}$.

We provide the following theorem to describe the feasibility of rank improvement by the distributed deployment.
\begin{theo} \label{theo:1.3.2}
There always exist multiple deployment locations of IRSs which can be carefully chosen to provide the DoF of $K$  by creating cascaded channel matrices $\tilde{\mathbf{H}}_i=\mathbf{H}_{i}^{\mathsf H} \mathbf{\Phi}_i^{\sf H} \mathbf G_{i}$.
\end{theo}

\begin{proof}
Let us consider an extreme case for any cascaded channel matrix $\tilde{\mathbf{H}}_i$ of rank one.
Due to the rank-one property, we can decompose $\tilde{\mathbf{H}}_i$ as $\tilde{\mathbf{H}}_i = \mathbf{a}_i \mathbf{b}_i^{\mathsf T}$.
Since array steering vectors $\mathbf{a}_i$ and $\mathbf{b}_i$ are determined by their particular angles, it is easy to find multiple array steering vectors with independent angles, thus leading to
\begin{align}
\sum_{i=1}^{I} \tilde{\mathbf{H}}_i  = \sum_{i=1}^{I} \mathbf{a}_i \mathbf{b}_i^{\mathsf T}
= \begin{bmatrix}
\mathbf{a}_1 & \mathbf{a}_2 & \cdots & \mathbf{a}_I 
\end{bmatrix}  
\begin{bmatrix}
\mathbf{b}_1^{\mathsf T} \\
\mathbf{b}_2^{\mathsf T} \\
\vdots \\
\mathbf{b}_I^{\mathsf T} 
\end{bmatrix}.
\label{eq:13}
\end{align}
For ease of illustration, we define $\mathbf{A}=[\mathbf{a}_1, \mathbf{a}_2, \cdots, \mathbf{a}_I]$ and $\mathbf{B}=[\mathbf{b}_1, \mathbf{b}_2, \cdots, \mathbf{b}_I]$.
If matrices $\mathbf{A}$ and $\mathbf{B}$ are constructed by independent angles, we have $\mathrm{rank}(\mathbf{A}) = I$ and $\mathrm{rank}(\mathbf{B}) = I$.
According to Sylvester's rank inequality, summing these cascaded channel matrices yields
\begin{align}
\mathrm{rank}(\sum_{i=1}^{I} \tilde{\mathbf{H}}_i) = 
\mathrm{rank}(\mathbf{A} \mathbf{B}^{\mathsf T})
\ge \mathrm{rank}(\mathbf{A})+\mathrm{rank}(\mathbf{B})-I
= I.
\label{eq:14}
\end{align}
This indicates that only $I \ge K$ IRSs deployed independently can provide the DoF of $K$ under the rank-one case. For the higher-rank case, it may require a small $I$.
\end{proof}

\section{IRS-empowered Slot Scheduling and Cost-efficient Reflection Optimization}\label{sec:5}
In this section, we study a slot-by-slot IRS refection design scheme to maximize the system sum rate. With the perfect CSI, a joint precoding and reflection pattern design problem is established, where the precoding and reflection pattern are optimized in each time slot. Then, we develop an AO framework to solve the optimization problem. Two optimization algorithms are proposed to solve the reflection pattern design problem. One is the SDR method in a closed-form update manner, and the other is the unit-circle manifold optimization (UCMO) method.
\subsection{Slot-by-slot Reflection Pattern Design}\label{sec:5.1}
To achieve the sum-rate maximization in the considered downlink mmWave network, the precoder at the BS and the reflection pattern of distributed IRSs should be jointly designed. 
The individual rate of the $k$-th MU is given by
\begin{equation}
R_k = \log_2(1+\gamma_k),
\label{eq:15}
\end{equation}
where the signal-to-interface-noise (SINR) of the $k$-th MU is computed by
\begin{equation}
\gamma_k = \frac{\left\vert \sum_{i=1}^{I}\mathbf h_{i,k}^{\mathsf H} \mathbf{\Phi}_i^{\mathsf H} \mathbf G_{i} \mathbf{w}_k  \right\vert^2}{\sum_{j=1,j\neq k}^{K} \left\vert \sum_{i=1}^{I}\mathbf h_{i,k}^{\mathsf H} \mathbf{\Phi}_i^{\mathsf H} \mathbf G_{i}  \mathbf{w}_j \right\vert^2 + \sigma^2}.
\label{eq:16}
\end{equation}
The resulting sum-rate maximization problem is formulated as
\begin{subequations}
\label{eq:17}
\begin{align}
\underset{\mathbf{W}, \{\mathbf{\Phi}\}_{i=1}^{I}}{\max} \quad & f_1(\mathbf{W}, \{\mathbf{\Phi}\}_{i=1}^{I} ) = \sum_{k=1}^{K} \omega_{k} R_{k}  \label{eq:17a}
\\
\mathrm{s.t.} \quad &\mathrm{tr}\left(\mathbf{W} \mathbf{W}^{\mathsf H}\right) \leq P_{\max},  \label{eq:17b}\\
&  \theta_{i,n} \in [0, 2\pi), \quad \forall i,\ \forall n, \label{eq:17c}
\end{align}
\end{subequations}
where the weight $\omega_{k}$ is the required service priority of the $k$-th MU.

It is observed that problem \eqref{eq:17} is a weighted sum-rate maximization problem. Due to the sum-logarithm form and complicated coupling relationship between $\mathbf{W}$ and $\{\mathbf{\Phi}\}_{i=1}^{I}$, it is difficult to optimize the precoder and reflection patterns simultaneously.
To address this issue, we examine the transformation steps for iteratively optimizing $\mathbf{W}$ and $\{\mathbf{\Phi}\}_{i=1}^{I}$.
First, we constitute an equivalent optimization problem to \eqref{eq:17} by introducing auxiliary variables $\boldsymbol \alpha = \left[\alpha_1, \alpha_2,\cdots, \alpha_K \right]^{\mathsf T}$ for SINR. Hence, the objective \eqref{eq:17a} can be transformed as
\begin{align}
f_2(\mathbf{W}, \{\mathbf{\Phi}\}_{i=1}^{I}, \boldsymbol \alpha) =  \underset{\mathbf{W}, \{\mathbf{\Phi}\}_{i=1}^{I}, \boldsymbol \alpha}{\max} \ \frac{1}{\ln 2} \sum_{k=1}^{K} \omega_{k} \ln(1+\alpha_k) -  \omega_{k} \alpha_k +  \frac{\omega_{k}(1+\alpha_k) \gamma_k}{1+\gamma_k}.
\label{eq:18}
\end{align}
The equivalence between \eqref{eq:18} and \eqref{eq:17} is confirmed by calculating the optimal $\alpha_k$ for $f_2$, given fixed $\mathbf{W}$ and $\{\mathbf{\Phi}\}_{i=1}^{I}$.
By setting $\frac{\partial f_2}{\partial {\alpha}_k}=0$, the optimal solution for $\alpha_k$ is ${\alpha}_k^{\mathrm{opt}} = \gamma_k$.
This means that optimizing the $\alpha_k$ equals to optimizing problem \eqref{eq:18}, while optimizing $\mathbf{W}$ and $\{\mathbf{\Phi}\}_{i=1}^{I}$ without need of dealing with the logarithm.

Based on the above discussions, given the fixed $\alpha_k$ in each iteration, only the last term in \eqref{eq:18} is involved in optimizing $\mathbf{W}$ and $\{\mathbf{\Phi}\}_{i=1}^{I}$ due to the SINR function of $\gamma_k$. From \eqref{eq:18}, the joint precoder and reflection pattern optimization problem is recast to
\begin{align}
\underset{\mathbf{W}, \{\mathbf{\Phi}\}_{i=1}^{I}}{\max} \quad & f_3(\mathbf{W}, \{\mathbf{\Phi}\}_{i=1}^{I})=\sum_{k=1}^{K} \frac{\omega_{k}(1+\alpha_k) \gamma_k}{1+\gamma_k}  \label{eq:19}
\\
\mathrm{s.t.} \quad & \eqref{eq:17b}, \eqref{eq:17c}. \nonumber
\end{align}
Now, we transform the joint precoder and reflection pattern optimization problem into a multi-ratio fractional programming (MRFP) form.

To solve the precoding matrix $\mathbf W$ when given fixed IRS reflection patterns $\{\mathbf{\Phi}_i\}_{i=1}^{I}$, we first define $\tilde{\mathbf h}_k^{\mathsf H} = \sum_{i=1}^{I}\mathbf h_{i,k}^{\mathsf H} \mathbf{\Phi}_i^{\mathsf H} \mathbf G_{i}$ for later convenience.
Substituting $\tilde{\mathbf h}_k$ into \eqref{eq:16}, problem \eqref{eq:19} can be rewritten as
\begin{align}
\underset{\mathbf{W}}{\max} \quad & f_4(\mathbf{W}) = \sum_{k=1}^{K} \frac{ \tilde{\omega}_{k} \left\vert  \tilde{\mathbf h}_k^{\mathsf H} \mathbf{w}_k  \right\vert^2}{\sum_{j=1}^{K} \left\vert \tilde{\mathbf h}_k^{\mathsf H}  \mathbf{w}_j \right\vert^2 + \sigma^2}
\label{eq:20}
\\
\mathrm{s.t.} \quad &  \eqref{eq:17b}, \nonumber
\end{align}
where $\tilde{\omega}_{k} = \omega_{k}(1+\alpha_k)$.
Problem \eqref{eq:20} is an MRFP problem. Based on the quadratic transform \cite{8314727}, we reformulate \eqref{eq:20} as
\begin{align}
f_5(\mathbf{W}, \boldsymbol{\beta}) =& \sum_{k=1}^{K}  2\sqrt{\tilde{\omega}_k} \mathsf{Re}\{ \beta_k^{\ast}\tilde{\mathbf h}_k^{\mathsf H} \mathbf{w}_k \} - \vert\beta_k\vert^2 \bigg(\sum_{j=1}^{K} \left\vert \tilde{\mathbf h}_k^{\mathsf H}  \mathbf{w}_j \right\vert^2 + \sigma^2 \bigg)  ,
\label{eq:21}
\end{align}
where $\boldsymbol{\beta} = [\beta_1, \beta_2, \cdots, \beta_K]^{\mathsf T}$ is a auxiliary vector by the quadratic transform. Next, we can directly obtain the optimal $\boldsymbol \beta$ by setting $\partial f_5 / \partial \beta_k = 0$.

Then, we adopt the Lagrange multiplier method to address the power constraint \eqref{eq:17a} to obtain the optimal $\mathbf w_k$, as given by
\begin{align}
{\mathbf w}_k^{\mathrm{opt}} = \sqrt{\tilde{\omega}_k} {\beta}_k \bigg(\mu \mathbf I_N + \sum_{i=1}^{K} \vert \beta_i \vert^2 \tilde{\mathbf h}_i \tilde{\mathbf h}_i^{\mathsf H} \bigg)^{-1} \tilde{\mathbf h}_k, 
\label{eq:22}
\end{align}
where $\mu \ge 0$ is the Lagrange multiplier.



To facilitate the optimization of reflection patterns, we adopt the following identity:
\begin{align}
\sum_{i=1}^{I}\mathbf h_{i,k}^{\mathsf H} \mathbf{\Phi}_i^{\mathsf H}  &= \sum_{i=1}^{I} \bm \theta_i^{\mathsf H} \mathrm{diag}\left( \mathbf h_{i,k}^{\mathsf H} \right) .
\label{eq:23}
\end{align}
Given $\boldsymbol{\alpha}$ and $\mathbf W$, substituting \eqref{eq:23} into \eqref{eq:19} yields
\begin{align}
\underset{\bm{\theta}_i}{\max} \quad & f_6(\bm{\theta}_i) = \sum_{k=1}^{K} \frac{\tilde{\omega}_k \left\vert \sum_{i=1}^{I} \bm \theta_i^{\mathsf H}  \mathbf v_{i,k,k} \right\vert^2}{\sum_{j=1}^{K} \left\vert \sum_{i=1}^{I}\bm \theta_i^{\mathsf H}  \mathbf v_{i,k,j} \right\vert^2 + \sigma^2}
\label{eq:24}
\\
\mathrm{s.t.} \quad & \eqref{eq:17c}, \nonumber
\end{align}
where $\mathbf v_{i,k,j} = \mathrm{diag}\left( \mathbf h_{i,k}^{\mathsf H}\right) \mathbf G_{i} \mathbf w_j$.
Problem \eqref{eq:24} can be further reduced to a compact form, as given by
\begin{align}
\underset{\tilde{\bm{\theta}} }{\max} \quad  f_7(\tilde{\bm{\theta}})
&= \sum_{k=1}^{K} \frac{\tilde{\omega}_k \left\vert {\tilde{\bm\theta}}^{\mathsf H} \tilde{\mathbf v}_{k,k} \right\vert^2}{\sum_{j=1}^{K} \left\vert {\tilde{\bm\theta}}^{\mathsf H} \tilde{\mathbf v}_{k,j} \right\vert^2 + \sigma^2}.
\label{eq:25}
\end{align}
where $\tilde{\bm\theta} = \mathrm{vec} \left(\left[\bm \theta_1, \bm \theta_2, \cdots, \bm \theta_I \right] \right)$ and $\tilde{\mathbf v}_{k,j}=\mathrm{vec}\left(\left[\mathbf v_{1,k,j}, \mathbf v_{2,k,j}, \cdots, \mathbf v_{I,k,j} \right] \right)$. 

Note that problem \eqref{eq:25} is also an MRFP, but with non-convex unit-modulus constraints \eqref{eq:17c}.
To get rid of the fractional form, applying the quadratic transform to \eqref{eq:25} yields
\begin{align}
\underset{\tilde{\bm\theta}, \bm{\rho}}{\max} \quad  f_8(\tilde{\bm\theta}, \bm\rho) &= \sum_{k=1}^{K} 2\sqrt{\tilde{\omega}_k} \mathsf{Re}\left\{ \rho_k^{\ast} \tilde{\bm\theta}^{\sf H} \tilde{\mathbf v}_{k,k} \right\} - \vert \rho_k \vert^2 \bigg(\sum_{j=1}^{K} \left\vert \tilde{\bm\theta}^{\sf H} \tilde{\mathbf v}_{k,j} \right\vert^2 + \sigma^2 \bigg) ,
\label{eq:26}
\end{align}
where $\rho_k$ refers to the auxiliary variable of quadratic transform. Then, the optimal $\rho_k$ can be directly obtained by setting $\frac{\partial f_8}{\rho_k}=0$.

Based on the Lagrange multiplier method, optimizing $\tilde{\bm\theta}$ reduces to
\begin{align}
\underset{\tilde{\bm\theta}}{\max} \quad  f_9(\tilde{\bm\theta}) &= -\tilde{\bm\theta}^{\mathsf H} \mathbf A \tilde{\bm\theta} + 2\mathsf{Re}\{\tilde{\bm\theta}^{\mathsf H} \mathbf b\} - \sum_{k=1}^{K} \vert \rho_k \vert^2 \sigma^2,
\label{eq:27}
\end{align}
where $\mathbf A = \sum_{k=1}^{K} \left( \vert \rho_k \vert^2 \sum_{j=1}^{K} \tilde{\mathbf v}_{k,j}\tilde{\mathbf v}_{k,j}^{\mathsf H}\right)$ and $\mathbf b = \sum_{k=1}^{K} \sqrt{\tilde{\omega}_k} \rho_k^{\ast} \tilde{\mathbf v}_{k,k}$.
Since $\mathbf A$ is a positive-definite matrix, problem \eqref{eq:27} becomes a quadratic constraint quadratic programming (QCQP).
To tackle the unit-modulus \eqref{eq:17c}, we relax it into a convex constraint $\vert\theta_n\vert\le 1, \forall n=1,2,\cdots, N$. Therefore, we rewrite \eqref{eq:27} as
\begin{subequations}
\label{eq:28}
\begin{align}
\underset{\tilde{\bm\theta}}{\max} \quad & f_{10}(\tilde{\bm\theta}) = -\tilde{\bm\theta}^{\mathsf H} \mathbf A \tilde{\bm\theta} + 2\mathsf{Re}\{\tilde{\bm\theta}^{\mathsf H} \mathbf b\}, \label{eq:28a} \\
\mathrm{s.t.} \quad &  \vert {\theta}_{n} \vert^2 \leq 1, \quad \forall n=1,2,\cdots, N, \label{eq:28b}
\end{align}
\end{subequations}
According to Slater's condition, it is easy to check that the strong duality holds. Hence, the Lagrangian of \eqref{eq:28} can be written as
\begin{align}
f_{\mathcal L}(\tilde{\bm\theta}, \bm\zeta) &= f_{10}(\tilde{\bm\theta}) - \sum_{n = 1}^{N} \zeta_{n} \left( \vert\theta_n\vert^2 - 1 \right) ,
\label{eq:29}
\end{align}
where $\zeta_k$ denotes the dual variable for constraints \eqref{eq:28b}. It is noticed that $f_{\mathcal L}(\tilde{\bm\theta}, \bm\zeta)$ is concave with respect to $\tilde{\bm\theta}$. Then, the dual transformation of \eqref{eq:29} is given by
\begin{subequations}
\label{eq:30}
\begin{align}
\underset{ \bm\zeta }{\min} \quad & f_{\mathcal D}(\bm\zeta) = \underset{\tilde{\bm\theta}}{\sup} \quad f_{\mathcal L}(\tilde{\bm\theta}, \bm\zeta), \label{eq:30a}\\
\mathrm{s.t.} \quad &  \zeta_{n} \geq 0, \quad \forall n=1,2,\cdots, N, \label{eq:30b}
\end{align}
\end{subequations}
where $f_{\mathcal D}(\bm\zeta)$ is the Lagrange dual function. To get the dual objective in the form that only depends on $\bm\zeta$, we attain the optimal $\tilde{\bm\theta}$ by setting ${\partial f_{\mathcal L}}/{\partial \tilde{\bm\theta}}=0$. Thus, we have
\begin{align}
\tilde{\bm\theta}^{\mathrm{opt}} &= \Big(\mathbf A + \mathrm{diag}\left(\bm\zeta\right) \Big)^{-1}  \mathbf b = \mathbf D(\bm\zeta) \mathbf b, \label{eq:31}
\end{align}
where $\mathbf D(\bm\zeta) = \left( \mathbf A + \mathrm{diag}\left(\bm\zeta\right) \right)^{-1}$. The optimal reflection patterns require the value of $\bm\zeta$. The following lemma reveals the property of optimal $\bm\zeta$.

\begin{lemm}\label{lemm:1.4.1}
The optimal ${\bm\zeta}$ for problem \eqref{eq:30} can be given by
\begin{align}
\bm\zeta^{\mathrm{opt}} = \left\{\zeta_{n} \geq 0: \left[\mathbf D(\bm\zeta) \mathbf b \right] \circ \left[\mathbf D(\bm\zeta) \mathbf b \right]^{\ast} = \mathbf{1} \right\}, \label{eq:32}
\end{align}
which is exactly equivalent to $\vert\theta_n\vert^2 = 1$, i.e., unit-modulus phase-shift constraints.
\end{lemm}

\begin{proof}
According to the chain rule in the matrix derivatives, we have
\begin{align}
\frac{\partial f_{\mathcal D} }{\partial \zeta_n} &= 1-\mathrm{tr}\left[ \mathbf D(\bm\zeta) \mathbf b \mathbf b^{\mathsf H} \mathbf D(\bm\zeta)  \frac{\partial \left( \mathbf A + \mathrm{diag}\left(\bm\zeta\right) \right)}{\partial \zeta_n} \right]  \nonumber \\
&= 1- \left[ \mathbf D(\bm\zeta) \mathbf b \mathbf b^{\mathsf H} \mathbf D(\bm\zeta) \right]_{n,n}. \label{eq:33}
\end{align}
Forcing \eqref{eq:33} to be zero yields
\begin{align}
\left[\mathbf D(\bm\zeta) \mathbf b \right] \odot \left[\mathbf D(\bm\zeta) \mathbf b \right]^{\ast} = \tilde{\bm\theta} \circ \tilde{\bm\theta}^{\mathsf H} = \mathbf 1 . \label{eq:34}
\end{align}
This completes our proof.
\end{proof}

Substituting \eqref{eq:31} into \eqref{eq:30}, we have
\begin{align}
f_{\mathcal D}(\bm\zeta) &= \mathbf b^{\sf H} \mathbf D(\bm\zeta) \mathbf b + \mathrm{tr}\left(\mathrm{diag}\left(\bm\zeta\right) \right). \label{eq:35}
\end{align}
Notably, based on the Schur complement, \eqref{eq:35} can be recast as an SDP form:
\begin{subequations}
\label{eq:36}
\begin{align}
\underset{ \bm\zeta,  \epsilon}{\max} \quad & \epsilon-\mathrm{tr}(\mathrm{diag}\left(\bm\zeta\right)), \\
\mathrm{s.t.} \quad &   
\begin{bmatrix}
\mathbf A + \mathrm{diag}\left(\bm\zeta\right) & \mathbf b \\
\mathbf b^{\mathsf H} & -\epsilon
\end{bmatrix} \succcurlyeq \mathbf 0,
\end{align}
\end{subequations}
which can be efficiently solved by convex tools such as CVX.

The details of the developed AO framework in this subsection are summarized in Algorithm \ref{alg:10.1}.

\begin{algorithm}[htbp]
\caption{The proposed AO framework.}
\label{alg:10.1}
\LinesNumbered
\SetKwInOut{KIN}{Initialization}
\KIN{Set feasible values of $\{\mathbf W^{(0)}, \mathbf \Phi_i^{(0)}\}$ and iteration index $t=0$.}
\Repeat{The function \eqref{eq:18} converges}{
Set $t \leftarrow t+1$\;
Update $\alpha_k^{(t)} = \gamma_k$ by \eqref{eq:16}\;
Update $\mathbf w_k^{(t)}$ by \eqref{eq:22}\;
Update $\bm\zeta^{(t)}$ by solving problem \eqref{eq:36} to obtain $\tilde{\bm\theta}^{(t)}$ by using \eqref{eq:31}\;
With given $\tilde{\bm\theta}^{(t)}$, update $\mathbf \Phi_i^{(t)}$\;
}
\end{algorithm}

We provide the following proposition to further illustrate the convergence of the proposed AO algorithm.

\begin{prop}\label{prop:1.4.1}
The problem in \eqref{eq:18} converges when the AO algorithm is adopted.
\end{prop}

\begin{proof}
For ease of illustration, we introduce a variable $t$ as the iteration index in Algorithm \ref{alg:10.1}. Since the optimum solution can be attained at each iteration, we have
\begin{align}
f_1(\mathbf W^{(t+1)}, \mathbf \Phi_i^{(t+1)}) &= f_2(\mathbf W^{(t+1)}, \mathbf \Phi_i^{(t+1)}, \bm \alpha^{(t+1)}) \nonumber \\
&\ge f_2(\mathbf W^{(t+1)}, \mathbf \Phi_i^{(t+1)}, \bm \alpha^{(t)}) \nonumber \\
&\ge f_2(\mathbf W^{(t)}, \mathbf \Phi_i^{(t+1)}, \bm \alpha^{(t)}) \nonumber \\
&\ge f_2(\mathbf W^{(t)}, \mathbf \Phi_i^{(t)}, \bm \alpha^{(t)}) \nonumber \\
&= f_1(\mathbf W^{(t)}, \mathbf \Phi_i^{(t)}).
\end{align}
Note that the weighted sum-rate maximization problem is bounded above due to the power constraints. Hence, the original objective $f_1$ is monotonically non-decreasing after each iteration.
\end{proof}

\subsection{Low-complexity Reflection Pattern Algorithm}
\label{sec:5.2}

Although the relaxation-based SDP method in Algorithm \ref{alg:10.1} is a straightforward solution to the reflection pattern optimization, unit-modulus constraints are generally converted into quadratic terms. By suppressing the rank-one constraint, problem \eqref{eq:28} becomes SDR form, typically solved by eigen-decomposition.
However, the SDR problem suffers from a critical drawback: the number of optimization variables increases quadratically with the number of IRS elements. For this reason, we develop a UCMO method that can directly deal with the unit-modulus constraints.

In light of the special geometry of the constraint $\vert \theta_{n} \vert = 1$, we resort to Riemannian-Geometric optimization tools \cite{AbsMahSep2008}. The feasible region of problem \eqref{eq:27} constitutes a unit-circle complex manifold. The manifold representation can provide a relatively concise form without relaxation.
The principle of Riemannian manifold optimization extends the gradient descent (GD) from Euclidean space to a Riemannian manifold space. 
Thus, the original constrained optimization problem can be transformed into an unconstrained optimization problem and optimized by Riemannian GD.

Referring to problem \eqref{eq:28}, the unit-modulus constrained optimization can be given by
\begin{subequations}
\label{eq:38}
\begin{align}
\underset{\tilde{\bm\theta}}{\max} \quad & f_{10}(\tilde{\bm\theta}) = -\tilde{\bm\theta}^{\mathsf H} \mathbf{A} \tilde{\bm\theta} + 2\mathsf{Re}\{\tilde{\bm\theta}^{\mathsf H} \mathbf{b}\}, \label{eq:38a} \\
\mathrm{s.t.} \quad &  \vert {\theta}_{n} \vert^2 = 1, \quad \forall n=1,2,\cdots, N. \label{eq:38b}
\end{align}
\end{subequations}
By constructing the smooth manifold structure on mapping space, problem \eqref{eq:38} can be reformulated as:
\begin{align}
\underset{ \tilde{\bm \theta} \in \mathcal S^{N}}{\min} \quad & f_{\mathrm{R}}( \tilde{\bm \theta}) = -{\tilde{\bm \theta}}^{\mathsf H} \mathbf{A} {\tilde{\bm \theta}} + 2 \mathsf{Re} \left( \tilde{\bm \theta}^{\mathsf H} \mathbf{b} \right), \label{eq:39}
\end{align}
where $\mathcal{S}^{N}$ denotes the manifold space defined by the unit-modulus constraints, i.e.,
\begin{align}
\mathcal S^{N} = \left\{ \tilde{\bm\theta} \in \mathbb{C}^{N \times 1}: \vert \theta_1 \vert = \vert \theta_2 \vert = \cdots = \vert \theta_{N} \vert = 1 \right\}, \label{eq:40}
\end{align}
where $\mathcal{S} = \left\{ \theta_n \in \mathbb{C}: \theta_n \theta_n^{\ast} = \mathsf{Re}\{\theta_n\}^2+\mathsf{Im}\{\theta_n\}^2 = 1  \right\}$ is known as a complex circle and can be viewed as a sub-manifold of $\mathbb{C}$. The search space $\mathcal{S}^{N}$ is the product of $N$ complex circles. Referred to as a unit-circle complex manifold, the search space $\mathcal{S}^N$ is a sub-manifold of $\mathbb{C}^{N \times 1}$.

\begin{figure}[t]
	\centering{}\includegraphics[width=3.5in]{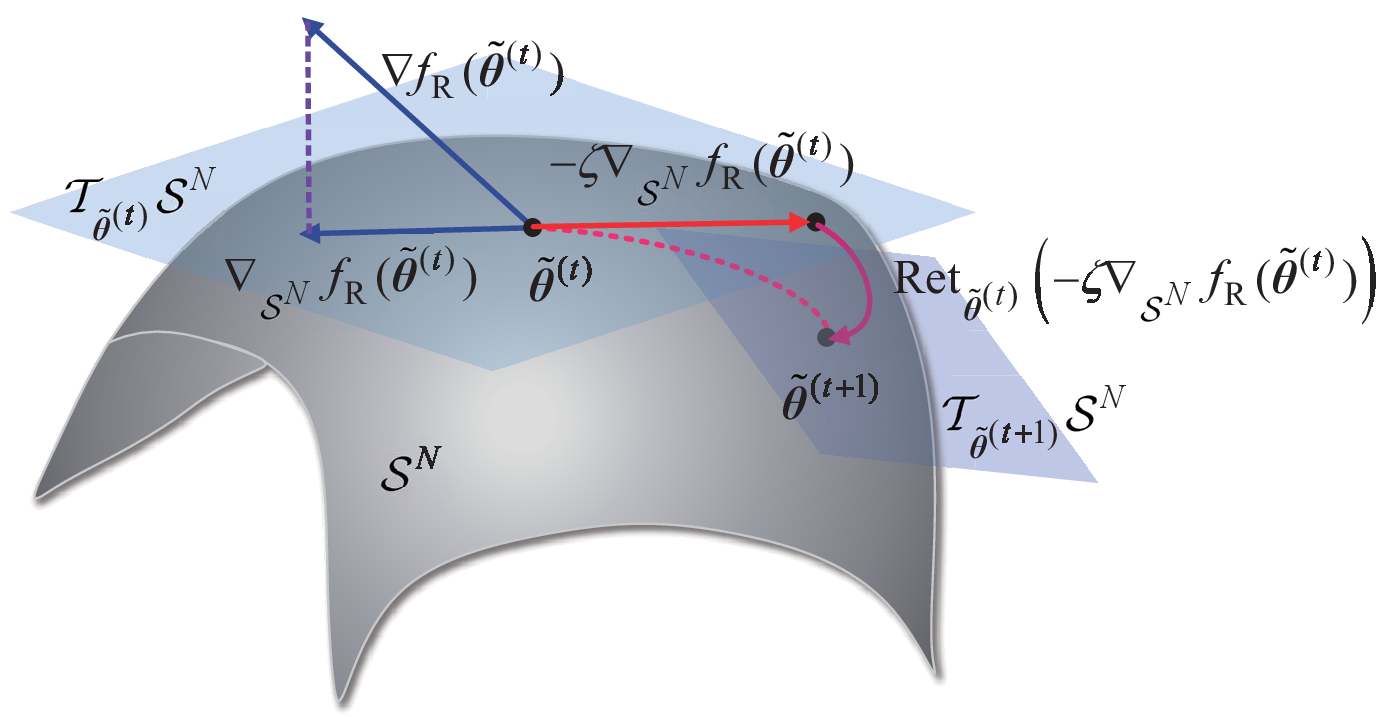}
	\caption{An illustration of GD on Riemannian manifolds.}
	\label{fig:7}
\end{figure}

The developed UCMO algorithm performs GD on the unit-circle complex manifold. The GD on the Riemannian manifold is similar to that in the Euclidean space. It consists of two phases: finding the descent direction of the current solution by computing the negative Riemannian gradient and decreasing the value of the objective function via line search~\cite{8125771}. 
The Riemannian gradient of $f_{\mathrm{R}}(\tilde{\bm\theta})$ at the current iteration point ${\tilde{\bm\theta}}^{(t)} \in \mathcal{S}^{N}$ is the projection of the search direction in the Euclidean space onto the tangent space $\mathcal{T}_{\tilde{\bm\theta}^{(t)}}\mathcal{S}^{N}$, as given by
\begin{align}
\mathcal{T}_{\tilde{\bm\theta}^{(t)}}\mathcal{S}^{N} = \left\{ \bm{\eta} \in \mathbb{C}^{(N+1) \times 1}: \mathsf{Re} \{ \bm{\eta}^{\ast} \circ {\tilde{\bm\theta}}^{(t)} \} = \mathbf{0}  \right\}. \label{eq:41}
\end{align}
Based on the projection operator, the Riemannian gradient of $f_{\mathrm{R}}(\tilde{\bm\theta}^{(t)})$ is obtained as
\begin{align}
\nabla_{\mathcal S^{N}} f_{\mathrm{R}}(\tilde{\bm\theta}^{(t)}) 
&= \mathsf{Proj}_{\mathcal{T}_{\tilde{\bm\theta}^{(t)}}\mathcal{S}^{N} } \left( \nabla f_{\mathrm{R}}(\tilde{\bm\theta}^{(t)}) \right) \nonumber\\
&= \nabla f_{\mathrm{R}}(\tilde{\bm\theta}^{(t)}) - \mathsf{Re} \{ \nabla f_{\mathrm{R}}(\tilde{\bm\theta}^{(t)})^{\ast} \circ f_{\mathrm{R}}(\tilde{\bm\theta}^{(t)}) \} \circ f_{\mathrm{R}}(\tilde{\bm\theta}^{(t)}), \label{eq:42}
\end{align}
where $\nabla f_{\mathrm{R}}({\tilde{\bm\theta}}^{(t)}) $ is the Euclidean gradient of $f_{\mathrm{R}}(\tilde{\bm\theta}^{(t)})$ at ${\tilde{\bm\theta}}^{(t)}$, as given by
\begin{align}
\nabla f_{\mathrm{R}}({\tilde{\bm\theta}}^{(t)}) = - 2 \mathbf{A} {{\tilde{\bm\theta}}^{(t)}} + 2 \mathbf{b}.  \label{eq:43}
\end{align}
As a result, the current point ${\tilde{\bm\theta}}^{(t)}$ in the tangent space $\mathcal{T}_{\tilde{\bm\theta}^{(t)}}\mathcal{S}^{N}$ is updated as
\begin{align}
{\tilde{\bm\theta}}'^{(t)} = {\tilde{\bm\theta}}^{(t)} - \zeta \nabla_{\mathcal S^{N}} f_{\mathrm{R}}(\tilde{\bm\theta}^{(t)}), \label{eq:44}
\end{align}
where $\zeta > 0$ is a carefully-chosen constant step size\footnote{To ensure stability and convergence of the UCMO algorithm, the step size $\zeta$ should be selected to satisfy $\zeta \le 1/\lambda_{\mathbf{A}}$ where $\lambda_{\mathbf{A}}$ represents the largest eigenvalue of the matrix $\mathbf{A}$ in problem \eqref{eq:39}. This optimization problem can be solved by leveraging the Manopt toolbox~\cite{manopt}.}. Since ${\tilde{\bm\theta}}'^{(t)}$ is still in the tangent space $\mathcal T_{{\bm \theta}^{(i)}} \mathcal{S}^{N+1}$ but may not be on the manifold $\mathcal S^{N}$, the Retraction mapping operation should be applied to move the point ${\tilde{\bm\theta}}'^{(t)}$ back onto the manifold $\mathcal{S}^{N}$. Finally, the point ${\tilde{\bm\theta}}^{(t+1)}$ after using the Retraction mapping is given by
\begin{align}
{\tilde{\bm\theta}}^{(t+1)} &= \mathsf{Ret}_{{\tilde{\bm\theta}}^{(t)} } \left( - \zeta \nabla_{\mathcal{S}^{N}} f_{\mathrm{R}}(\tilde{\bm\theta}^{(t)}) \right) = \frac{{\tilde{\bm\theta}}^{(t)} - \zeta \nabla_{\mathcal S^{N}} f_{\mathrm{R}}(\tilde{\bm\theta}^{(t)})}{\Vert {\tilde{\bm\theta}}^{(t)} - \zeta \nabla_{\mathcal S^{N}} f_{\mathrm{R}}(\tilde{\bm\theta}^{(t)}) \Vert}. 
\label{eq:45}
\end{align}
The above operations are illustrated in Fig. \ref{fig:7} and summarized in Algorithm \ref{alg:2}.

\begin{algorithm}[htbp]
\caption{The proposed UCMO algorithm for reflection pattern optimization.}
\label{alg:2}
\LinesNumbered
\SetKwInOut{KIN}{Initialization}
\SetKwInOut{KOUT}{Output}
\KIN{Set feasible values of $\tilde{\bm\theta}^{(0)}$ and iteration index $t=0$.}
\KOUT{Reflection pattern $\tilde{\bm\theta}$.}
\Repeat{The value of $\vert f_{\mathrm{R}}(\tilde{\bm\theta}^{(t)}) - f_{\mathrm{R}}(\tilde{\bm\theta}^{(t-1)}) \vert$ in \eqref{eq:39} converges}{
Set $t \leftarrow t+1$\;
Calculate the Euclidean gradient $\nabla f_{\mathrm{R}}(\tilde{\bm\theta}^{(t)})$ at $\tilde{\bm\theta}^{(t)}$ using \eqref{eq:43}\;
Construct the tangent space $\mathcal T_{{\bm\theta}^{(t)}} \mathcal{S}^{N}$ and calculate the current Riemannian gradient $\nabla_{\mathcal S^{N}} f_{\mathrm{R}}(\tilde{\bm\theta}^{(t)})$ using \eqref{eq:42}\;
Perform GD algorithm over the current tangent space using \eqref{eq:44}\;
Update ${\tilde{\bm\theta}}^{(t+1)}$ using the Retraction mapping operator according to \eqref{eq:45}\;
}
\end{algorithm}

We compare the computational complexity between the SDR and the developed UCMO algorithms in reflection pattern optimization. The SDR presented in Section \ref{sec:5.1} incurs a high complexity of $\mathcal{O}(N^{6.5})$ \cite{cmu2.12095} per iteration, since it expands the variable dimension. For the proposed UCMO algorithm described in Algorithm \ref{alg:2}, the complexity is dominated by the evaluation of the Riemannian GD, which is $\mathcal{O}(N^{2})$ per iteration \cite{8706630}. To this end, the proposed UCMO algorithm is order of magnitude lower than the SDR-based alternative.

\section{Two-timescale Reflection Pattern Design}\label{sec:6}
Precise CSI, real-time reflection pattern computing and control, and low-overhead configuration signaling are key challenges for IRS optimization. Therefore, in this section, we develop a two-timescale reflection pattern scheme for the IRS-assisted MIMO system to maximize the achievable average sum-rate, which configures the IRS per frame based on the S-CSI and updates the active beamforming according to the available I-CSI slot by slot during the frame. Furthermore, a low-complexity PSO method, which is inspired by the data-driven paradigm, is efficiently developed to solve the large-timescale reflection optimization.

\subsection{Temporally Correlated Channel and Two-timescale Transmission Strategy} \label{sec:6.1}
We consider a generic IRS-assisted MIMO communication system, where the BS is equipped with $N_{\mathrm{t}}$ transmit antennas of ULA, and MU is equipped with $K$ receive antennas of ULA.
The IRS is a uniform rectangular array (URA) consisting of $N=N_x N_y$ meta-atoms, with $N_x$ meta-atoms per row and $N_y$ meta-atoms per column.
The direct channel from the BS to the MU is denoted by $\mathbf{H}_{\mathrm{d}}$, and the BS-IRS and IRS-MU channels are denoted by $\mathbf{G}$ and $\mathbf{H}_{\mathrm{r}}$, respectively. For the brevity of notation, we define $\mathcal{H}=\{\mathbf{G}, \mathbf{H}_{\mathrm{d}}, \mathbf{H}_{\mathrm{r}}\}$ as the full channel ensemble.

Typically, IRS is deployed to achieve the LoS propagation between the BS and the IRS. So, the rank-one BS-IRS channel is assumed. Generally, the locations of IRSs are known to the BS, and so is the BS-IRS channel, which is thus assumed to be time-invariant due to the stationarity of both the BS and IRS.
The BS-IRS channel is given by
\begin{align}
\mathbf{G}= \sqrt{L_{\mathrm{br}}} \mathbf{a}(\phi) \mathbf{b}^{\mathsf{H}}(\phi_{\mathrm{e}}, \phi_{\mathrm{a}}) \in \mathbb{C}^{N_{\mathrm{t}} \times N},
\label{eq:46}
\end{align}
where $L_{\mathrm{br}}$ is large-scale fading, $\phi$ is the angle-of-departure (AoD) at the BS, $\phi_{\mathrm{e}}(\text{or} \ \phi_{\mathrm{a}})$ is the elevation (or azimuth) angles-of-arrival (AoA) at the IRS, $\mathbf{a}(\phi)$ is the array steering vector of the BS, and $\mathbf{b}(\phi_{\mathrm{e}}, \phi_{\mathrm{a}})$ is the array steering vector of the IRS. The array steering vectors are defined in the same way as \eqref{eq:3}.

Due to the mobility of MUs, the IRS-MU and BS-MU channels are multi-path and fast time-varying.
To illustrate,  we focus on a time frame within which the S-CSI of all links remains unchanged. The time frame is divided evenly into $T$ time slots. The channel coefficients of the links, i.e., I-CSI, do not change during a time slot. The NLoS channel coefficient matrices are correlated across different time slots.
At time slot $t$, the IRS-MU channel can be given
\begin{align}
\mathbf{H}_{\mathrm{r}} [t] = \sqrt{L_{\mathrm{ur}}} \left(  \kappa \mathbf{H}_{\mathrm{r}}^{\mathrm{LoS}} + \sqrt{1-\kappa^2} \mathbf{H}_{\mathrm{r}}^{\mathrm{NLoS}} [t] \right)\in \mathbb{C}^{N \times K},
\label{eq:47}
\end{align}
where $L_{\mathrm{ur}}$ accounts for large-scale fading; $\kappa$ is the Rician fading-related factor; $\mathbf{H}_{\mathrm{r}}^{\mathrm{LoS}}$ and $\mathbf{H}_{\mathrm{r}}^{\mathrm{NLoS}}\sim \mathcal{CN}(0,\mathbf{I})$ are the deterministic LoS component and the NLoS component, respectively.
Likewise, the BS-UE channel can be given by
\begin{align}
\mathbf{H}_{\mathrm{d}}[t] = \sqrt{L_{\mathrm{bu}}} \left( \kappa \mathbf{H}_{\mathrm{d}}^{\mathrm{LoS}} + \sqrt{1-\kappa^2} \mathbf{H}_{\mathrm{d}}^{\mathrm{NLoS}} [t] \right)\in \mathbb{C}^{N_{\mathrm{t}} \times K}.
\label{eq:48}
\end{align}

Herein, we use the first-order auto-regressive process \cite{6937178} to characterize temporal-variant NLoS channel as
\begin{align}
\mathbf{H}^{\mathrm{NLoS}} [t] = \rho \mathbf{H}^{\mathrm{NLoS}} [t-\tau] + \mathbf{\Delta}[t],
\label{eq:49}
\end{align}
where $\mathbf{\Delta}[t] \sim \mathcal{CN}\left(0, (1-\rho^2 )\mathbf{I} \right)$ yields the complex Gaussian distribution. Following the Jakes' model \cite{8753608}, the temporal correlation coefficient is $\rho=J_0(2\pi\bar{f}_d)$, where $J_0(\cdot)$ is the zero-th order Bessel function of the first kind; $\bar{f}_d = f_d \tau$ is the normalized Doppler frequency; $f_d$ is the maximum Doppler frequency; and $\tau$ is the delay between the uplink channel estimation and downlink transmission.

In time-varying scenarios, the BS first estimates the CSI based on the uplink pilots sent by the MU. Then, according to the estimated CSI, the BS computes the beamformer for downlink transmissions. When used for downlink transmissions, the CSI estimated in the uplink can expire by $\tau$ slots. Thus, the channel reciprocity may not hold in a practical TDD system.

In the considered IRS-assisted MIMO system, the BS transmits $M \le N_{\mathrm{t}}$ data streams to the MU.
At the $t$-th slot, the received signal at the MU is given by
\begin{align}
\mathbf{y}[t]
&= \mathbf{W}_{\mathrm{r}}^{\mathsf{H}}    \underbrace{ \left( \mathbf{H}_{\mathrm{d}}^{\mathsf{H}} [t] + \mathbf{H}_{\mathrm{r}}^{\mathsf{H}}[t] \mathbf{\Theta}^{\mathsf{H}}[t]  \mathbf{G}^{\mathsf{H}} [t] \right) }_{\text{effective channel}} \mathbf{W}_{\mathrm{t}} \mathbf{P}[t] \mathbf{x}[t]  + \mathbf{n}[t]    \nonumber \\
&= \mathbf{W}_{\mathrm{r}}^{\mathsf{H}}\tilde{\mathbf{H}}^{\mathsf{H}} [t]  \mathbf{W}_{\mathrm{t}} \mathbf{P}[t]   \mathbf{x}[t]  + \mathbf{n}[t] ,
\label{eq:50}
\end{align}
where $\mathbf{x} \in \mathbb{C}^{M \times 1}$ is the transmit signal, $\mathbf{n} \sim \mathcal{CN}(0, \sigma^2 \mathbf{I}_{K})$ is the additive Gaussian noise, $\mathbf{W}_{\mathrm{t}}$ is the transmit beamformer at the BS, $\mathbf{W}_{\mathrm{r}}$ is the receive beamformer at the MU, $\mathbf{P}=\mathrm{diag}\left\{[ \sqrt{P_1}, \sqrt{P_2},\cdots, \sqrt{P_{M}} ]\right\} \in \mathbb{C}^{M \times M}$ collects the transmit powers allocated to data streams, $\mathbf{\Theta} =\mathrm{diag} \left\{\exp (j\boldsymbol{\theta}) \right\} \in \mathbb{C}^{N\times N}$ is the reflection pattern matrix of the IRS, and $\tilde{\mathbf{H}} \triangleq \mathbf{H}_{\mathrm{d}}  + \mathbf{G} \mathbf{\Theta} \mathbf{H}_{\mathrm{r}} \in \mathbb{C}^{N_{\mathrm{t}} \times K}$ is defined as \emph{effective channel}.

It is not straightforward to estimate the effective channel $\tilde{\mathbf{H}}$, because part of it can be augmented by the IRS. The passive meta-atoms are made up of PIN-diodes and do not involve active RF chains or any signal processing. For this reason, cascaded channel estimation is widely adopted to estimate the effective channels~\cite{9133156}.
However, as for the slot-by-slot reflection configuration of the IRS, the cascaded channel estimation needs to be frequently conducted every time the IRS is reconfigured.
In consequence, unacceptably large training overhead and delay would be required.

To this end, we develop a two-timescale beamforming strategy, configuring the IRS once per frame based on the S-CSI and updating the active beamformers in every slot based on available CSI, which may be outdated.
As illustrated in Fig. \ref{fig:8}, the S-CSI parameters including $\{ L_{\mathrm{ur}}, L_{\mathrm{bu}} \}$ and $\mathbf{H}_{\mathrm{r}/\mathrm{d}}^{\mathrm{LoS}}$ are evaluated at the beginning of a frame based on the S-CSI estimation approach~\cite{9369969, 9408385}. Based on the S-CSI, the smart controller optimizes the IRS reflection patterns for the entire frame.
Given the large-timescale IRS reflection pattern in a frame, the small-timescale beamformers can be conducted slot by slot throughout the frame.

\begin{figure}[t]
	\centering{}\includegraphics[scale=0.65]{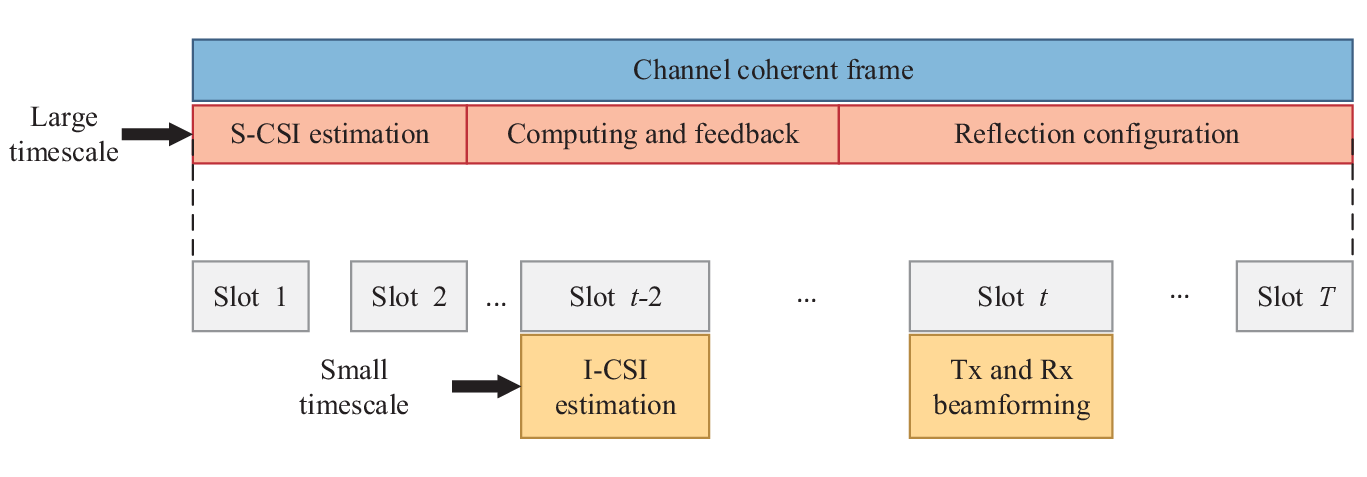}
	\caption{Pictorial illustration of the proposed two-timescale beamforming strategy.}
	\label{fig:8}
\end{figure}

Next, we analyze the small-timescale beamforming schemes in the case of time-varying channels. 
As shown in Fig. \ref{fig:9}(b), both BS and MU directly apply singular value decomposition (SVD) beamformers generated from the O-CSI, i.e., $\tilde{\mathbf{H}}[t-\tau]=\mathbf{U}[t-\tau]\mathbf{\Lambda}[t-\tau]\mathbf{V}^{\mathsf H}[t-\tau]$. The mismatch between beamformers and actual channel ($\mathbf{V}[t-\tau]$ and $\tilde{\mathbf{H}}[t]$) can lead to non-negligible inter-stream interference, deteriorating the rate performance. This is because mismatched unitary matrices cannot fully parallelize subchannels.

For this reason, we adopt a hybrid SVD and zero-forcing (ZF) beamforming scheme and extend it under power constraints to optimize the power allocation based on the O-CSI. 
As shown in Fig. \ref{fig:9}(c), the BS still performs the SVD beamformer $\mathbf{W}_{\mathrm{t}}=\mathbf{V}[t-\tau]$ based on the O-CSI, but sends the downlink pilots before the transmit beamformer, instead of sending pilots before the wireless channel after the transmit beamformer. Then, the MU conducts the ZF beamformer based on the combined information of the channel and transmit beamformer in the current slot. 

Since both pilots and symbols are precoded and transmitted through the equivalent combined channel, $\check{\mathbf{H}}[t] \triangleq \mathbf{V}^{\mathsf{H}} [t-1]\tilde{\mathbf{H}}[t] \in \mathbb{C}^{M \times K}$. The MU can estimate $\check{\mathbf{H}}[t]$ based on the precoded pilots.	
Based on the estimated combined channel, i.e., $\check{\mathbf{H}}$, the MU detects the data symbols using the following ZF beamformer:
\begin{align}
	\mathbf{W}_{\mathrm{r}}
	&= \left[ \Big(  \tilde{\mathbf{H}}^{\mathsf{H}}[t] \mathbf{V}[t-\tau] \Big)^{\mathsf{H}}  \tilde{\mathbf{H}}^{\mathsf{H}}[t] \mathbf{V}[t-\tau]   \right]^{-1} \Big(  \tilde{\mathbf{H}}^{\mathsf{H}} [t]\mathbf{V}[t-\tau] \Big)^{\mathsf{H}} \nonumber \\
	&= \left( \check{\mathbf{H}}[t] \check{\mathbf{H}}^{\mathsf{H}}[t] \right)^{-1}  \check{\mathbf{H}}[t] . 
 \label{eq:51}
\end{align}
The received signal is obtained as
\begin{align}
	\mathbf{y}[t] &= \mathbf{W}_{\mathrm{r}}^{\mathsf{H}} \left( \tilde{\mathbf{H}}^{\mathsf{H}}[t] \mathbf{V}[t-\tau] \mathbf{P} [t] \mathbf{x} [t]+ \mathbf{n}[t] \right) \nonumber\\
	&=  \mathbf{P}[t] \mathbf{x}[t] + \left( \check{\mathbf{H}}[t] \check{\mathbf{H}}^{\mathsf{H}}[t] \right)^{-1}  \check{\mathbf{H}}[t]  \mathbf{n} [t].
 \label{eq:52}
\end{align}
The resulting SINR of the $m$-th data stream is given by
\begin{align}
	\gamma_m[t]
	&=\frac{ P_m [t] }{\sigma^2 \left[ \left( \check{\mathbf{H}}[t] \check{\mathbf{H}}^{\mathsf{H}}[t] \right)^{-1}  \right]_{(m,m)} }. 
 \label{eq:53}
\end{align}


\begin{figure}[!h]
	\centering{}\includegraphics[scale=0.4]{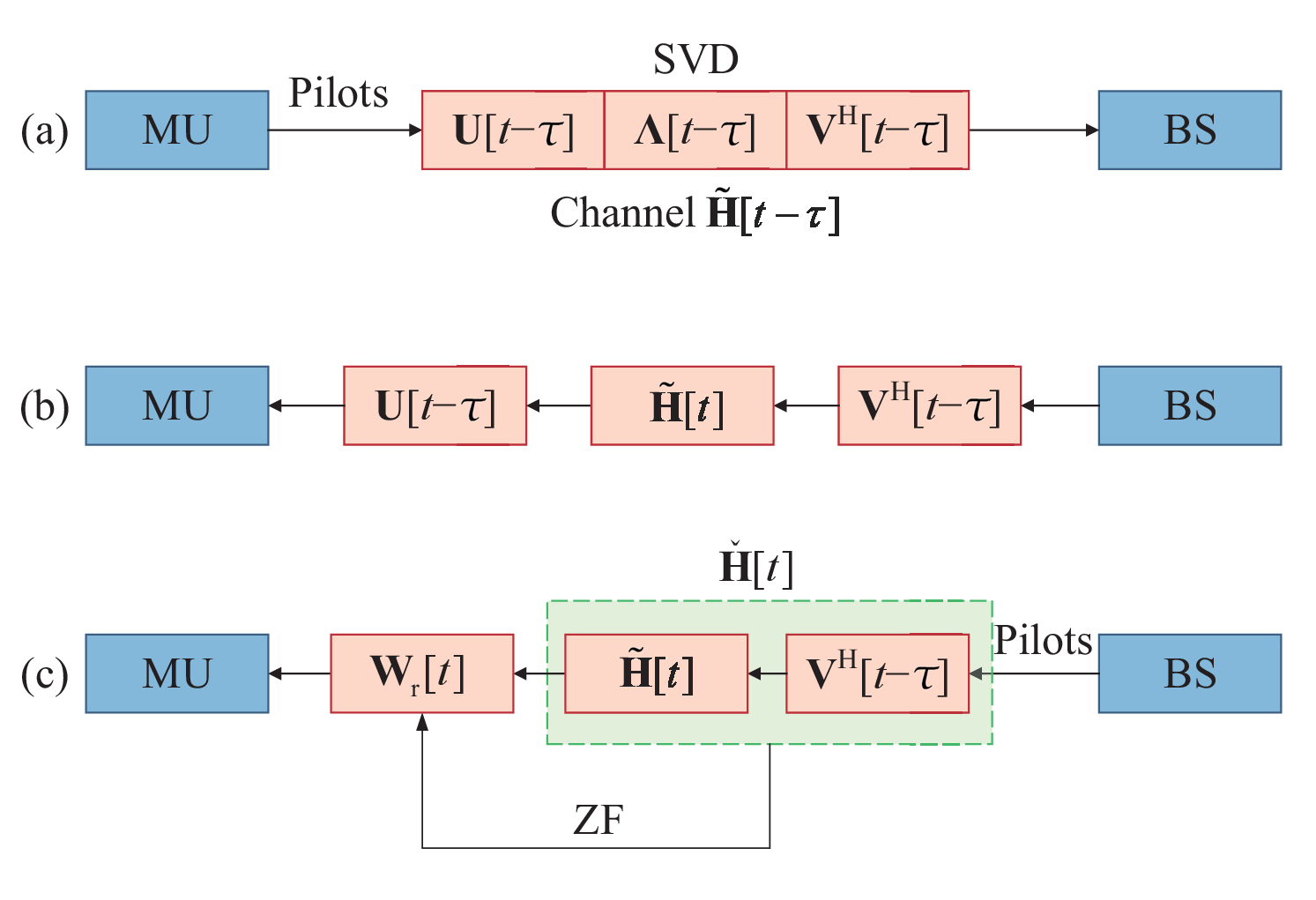}
	\caption{The small-timescale beamforming procedure: (a) uplink channel estimation; (b) conventional SVD-based downlink transmission; and (c) hybrid SVD and ZF beamforming-based downlink transmission.}
	\label{fig:9}
\end{figure}

This beamforming scheme can efficiently eliminate the inter-stream interference resulting from the SVD beamforming based on the O-CSI.
The SVD precoding with water-filling power allocation is also known to be optimal under I-CSI~\cite{eldad2008next}. It is still likely to produce effective parallel streams based on the O-CSI by taking advantage of the temporal correlation of the auto-regressive channels. 
To this end, the small-timescale SVD-ZF beamforming strategy is increasingly closer to the optimum as the channel changes more slowly and undergoes stronger temporal correlation. The strategy becomes optimal when the channel stops changing.

With the hybrid SVD and ZF scheme in the small-timescale, the goal of the proposed two-timescale transmission strategy is to maximize the AASR of the considered IRS-assisted MIMO system by jointly optimizing the long-term reflection pattern of the IRS and the per-slot power allocation of the BS.
This problem is formulated as
\begin{subequations}
\begin{align}
\text{(P1)}: \quad & \underset{ \mathbf{\Theta} }{\max} \  \mathbb{E} \left\{ \underset{ \mathbf{P} }{\max} \  \sum_{m=1}^{M}
\log_2 \Big[ 1+ \gamma_m (\mathbf{\Theta}, \mathbf{P}[t]) \Big]  \right\} \label{eq:54} \\
\mathrm{s.t.} \quad &   \sum_{m=1}^{M} P_m  = P_{\mathrm{tot}},  \quad P_m \ge 0 , \forall m, \label{eq:54a} \\
& \left\vert [\mathbf{\Theta}]_{n,n} \right\vert = 1, \quad  \mathrm{angle}\left( [\mathbf{\Theta}]_{n,n} \right) \in [0, 2\pi), \forall n. \label{eq:54b}
\end{align}
\end{subequations}
where \eqref{eq:54a} guarantees the total transmit power of the BS does not exceed its power budget $P_{\mathrm{tot}}$, and \eqref{eq:54b} specifies the unit-modulus nature and the phase-shift range of the IRS meta-atoms.

\subsection{Reflection Pattern Design with Statistical CSI} \label{sec:6.2}

We reformulate problem (P1) as a joint design of power allocation at the BS and reflection configuration at the IRS. Then, we provide the solutions for large- and small-timescale designs.
We first rewrite the SINR of the $m$-th data stream at slot $t$, $\gamma_m(P_m[t], \mathbf{\Theta}, \mathcal{H}[t-\tau])$, as
\begin{align}
 \gamma_m(P_m[t], \mathbf{\Theta}, \mathcal{H}[t-\tau]) =&  \frac{ P_m[t] }{\sigma^2 \left[ \left\{ \check{\mathbf{H}}[t] \check{\mathbf{H}}^{\mathsf{H}} [t]\right\}^{-1} \right]_{(m, m)} } \nonumber\\
=&  \frac{ P_m[t] }{\sigma^2 \left[\mathbf{R}^{-1}[t] \right]_{(m,m)} }  = \frac{ P_m[t] }{\sigma^2 f_m( \boldsymbol{\theta} )[t] } , \label{eq:55}
\end{align}
where $\mathbf{R}[t]=\check{\mathbf{H}}[t] \check{\mathbf{H}}^{\mathsf{H}} [t]$ is the channel covariance, $f_m( \boldsymbol{\theta} )[t]=\left[  \mathbf{R}^{-1} [t] \right]_{(m,m)}$ is the $m$-th diagonal entry of $\mathbf{R}^{-1}$.
Problem (P1) can be rewritten as
\begin{subequations}
\begin{align}
\text{(P1')}: \quad & \underset{ \boldsymbol{\theta} }{\max} \  \mathbb{E} \left\{ \underset{ \mathbf{p} }{\max} \  \sum_{m=1}^{M}
\log_2 \left[ 1+ \frac{P_m [t]}{\sigma^2 f_m( \boldsymbol{\theta} ) [t]} \right]  \right\} \label{eq:56} \\
\mathrm{s.t.} \quad &  \mathbf{p} \succeq 0, \quad \mathbf{1}^{\mathsf{T}}\mathbf{p} = P_{\mathrm{tot}}, \label{eq:56a} \\
&  \mathbf{\Theta}= \mathrm{diag}\left\{ e^{ j\boldsymbol{\theta} } \right\},  {\theta}_{n} \in [0, 2\pi), \forall n=1,2,\cdots, N,  \label{eq:56b}
\end{align}
\end{subequations}
where $\mathbf{p}=[P_1, \cdots, P_M]^{\mathsf{T}}$ collects the transmit powers of $M$ data streams.

Problem (P1') is a stochastic optimization problem because of the per-slot optimization of the power allocation at the small timescale.
At each slot, the power allocation is based on the available O-CSI, with the constant IRS reflection pattern applied at the current frame.
On the other hand, the optimization of the IRS configuration needs to farsightedly maximize the AASR over random channel samples and potential per-slot power allocations throughout the frame.
Although the stochastic optimization problem is converted into a deterministic non-convex optimization problem, minimizing each diagonal entry of the inverse of the channel covariance, i.e., $\left[ \mathbf{R}^{-1}\right]_{(m,m)}$, is still arduous.

Conversely, the small-timescale power allocation is a deterministic optimization problem.
Given the IRS reflection pattern of the large timescale, with the hybrid SVD and ZF scheme,  problem (P1) is reformulated as
\begin{align}
	\text{(P2)}: \quad 
	\underset{ \mathbf{p} }{\max} \  \sum_{m=1}^{M}
	\log \left( 1+ \frac{P_m[t] f_m^{-1}(\boldsymbol{\theta}) [t] }{\sigma^2 } \right) \quad
	\mathrm{s.t.} \    \eqref{eq:56a}  \nonumber.
\end{align}
Problem (P2) is convex with respect to $\mathbf{p}$, and can be solved by the classic water-filling algorithm. The optimal small-timescale power allocation is given by
\begin{align}
	P_m^{\mathrm{opt}} [t]=  \left[ \frac{1}{\mu \ln 2}- \sigma^2 f_m(\boldsymbol{\theta}) [t] \right]^{+}. \label{eq:57}
\end{align}
where $\mu$ is the Lagrange multiplier and $[x]^{+}= \max\{x, 0\}$, 
The resulting per-slot achievable rate is
\begin{align}
	\sum_{m=1}^{M}
	\log_2 \left( 1+ \frac{P_m [t]}{\sigma^2 f_m(\boldsymbol{\theta})[t] } \right)
	= \sum_{m=1}^{M}
	\log \left( 1+ \frac{1}{\sigma^2} \left[ \frac{ 1 }{ \mu f_m(\boldsymbol{\theta}) [t] \ln 2 } - \sigma^2\right]^{+} \right).
 \label{eq:58}
\end{align}

The per-frame IRS reflection optimization is a stochastic optimization problem due to the expectation operation in the objective of the problem (P1').
To circumvent this impasse, we first transform the objective in a deterministic form by taking the average rate of $B$ random channel samples to approximate the expectation in the objective.
Suppose that the $b$-th sample of the IRS-related time-variant channels in the $t$-th time slot is $\mathcal{H}^{b}[t]= \{ \mathbf{H}_{\mathrm{d}}^{b}[t-\tau], \mathbf{H}_{\mathrm{r}}^{b}[t-\tau] \}$, accounting for the O-CSI at the BS. The AASR of the $m$-th data-stream is approximated by
\begin{align}
\bar{r}_m \left(P_m^{\mathrm{opt}}( \boldsymbol{\theta}, \{\mathcal{H}^{b}\}_{b=1}^{B}), \boldsymbol{\theta} \right) = \frac{\sum_{t=1}^{T} \sum_{b=1}^{B} \log_2 \left[ 1 +  \gamma_m \left( P_m[t], \boldsymbol{\theta}, \mathcal{H}^{b}[t-\tau] \right)  \right]}{TB}. \label{eq:59}
\end{align}

The large-timescale passive beamforming problem is cast as
\begin{align}
\text{(P3)}: \quad   \underset{ \boldsymbol{\theta} }{\max} \  \sum_{m=1}^{M} \bar{r}_m \left(P_m^{\mathrm{opt}}( \boldsymbol{\theta}, \{\mathcal{H}^{b}\}_{b=1}^{B}), \boldsymbol{\theta} \right)  \quad
\mathrm{s.t.} \  \eqref{eq:56b}, \nonumber
\end{align}
Apart from the non-convex constant-modulus constraint, the power allocation variables are closely coupled with the reflection patterns of the IRS. Moreover, maximizing the AASR involves minimizing the diagonal entries, i.e., $\left[\mathbf{R}^{-1}\right]_{(m,m)}$, which makes it difficult to optimize $\bm{\theta}$ since $\left[\mathbf{R}^{-1}\right]_{(m,m)}$ cannot be written as an explicit function of $\bm{\theta}$.
Conventional convex optimization or convexification methods cannot be applied to solve the problem (P3).

We propose a modified PSO framework to solve the problem (P3).
As a widely-used meta-heuristic bionic optimization algorithm~\cite{jian2020hybrid}, PSO can find the optimal solution with little prior information. It enjoys the benefits of fewer parameters, simple calculation implementation, and fast convergence.
In PSO, the position of a particle stands for a potential solution. The fitness function of a particle is typically defined to be the optimization objective.
We define a possible reflection vector $\boldsymbol{\theta}$ as the position of a particle.
The fitness function can be taken as the objective of (P3), as given by
\begin{align}
J(\boldsymbol{\theta}) = \sum_{m=1}^{M} \bar{r}_m \left(P_m^{\mathrm{opt}}( \boldsymbol{\theta}, \{\mathcal{H}^{b}\}_{b=1}^{B} ), \boldsymbol{\theta} \right).
\label{eq:60}
\end{align}
For each particle $\boldsymbol{\theta}$, the optimal power solution can be uniquely found according to \eqref{eq:57}. This guarantees the uniqueness of the fitness function in \eqref{eq:60}.
There are $P$ particles utilized to search for the optimal reflection pattern in the constrained search space, and each is assigned a velocity at every iteration.
The set of particle positions is denoted by $\mathcal{L}_P = \{\boldsymbol{\theta}_1, \boldsymbol{\theta}_2, \cdots, \boldsymbol{\theta}_P\}$.
A particle is assessed based on its fitness function in every iteration to determine whether the current position implies a good solution. The particles record their best positions ever found.

The global optimal position is selected from the best positions of all particles.
Let $\hat{\boldsymbol{\theta}}_{p}^{(i)}$ and $\hat{\boldsymbol{\theta}}^{(i)}$ denote the best position of the $p$-th particle and the global optimal position of all particles at the $i$-th iteration, respectively.
The $p$-th particle updates its position according to its velocity $\mathbf{v}_{p}^{(i)}$ at the $i$-th iteration.
The newly updated position is used to determine the best position of this particle, if the fitness value of the previous best position is lower than that of the new position. Otherwise, the best position of this particle remains unchanged. After one round of the fitness value evaluation, the global optimal position of all particles is accordingly determined.
Both the velocity and position of the $p$-th particle at the $i$-th iteration are updated by
\begin{align}
\mathbf{v}_{p}^{(i)} &= w \mathbf{v}_{p}^{(i-1)} + c_1 \varepsilon_1 (\hat{\boldsymbol{\theta}}_{p}^{(i-1)} - \boldsymbol{\theta}_{p}^{(i-1)}) + c_2 \varepsilon_2 (\hat{\boldsymbol{\theta}}^{(i-1)} - \boldsymbol{\theta}_{p}^{(i-1)}), \label{eq:61}\\
\boldsymbol{\theta}_{p}^{(i)} &= \boldsymbol{\theta}_{p}^{(i-1)} + \mathbf{v}_{p}^{(i)}, \label{eq:62}
\end{align}
where $w$ is the non-negative inertia weight of a particle; $c_1$ and $c_2$ denote cognitive and social scaling factors, respectively; and $\varepsilon_1$ and $\varepsilon_2$ are two independent random variables with a uniform distribution in $(0,1)$. The other control parameters, such as the swarm size $P$ and iteration number $I_{\mathrm{iter}}$, are critical for the convergence rate and global search performance.
Once the parameters updated in \eqref{eq:62} are outside $(-\pi, \pi]$, their boundary values are taken.

Each particle requires a notable amount of calculations for fitness evaluation per iteration since the fitness function requires a large number of random channel samples.
It can be computationally prohibitive to take large-scale channel samples (which are four-dimensional matrices with size of $N_{\mathrm{t}}\times K \times B \times P$) to obtain a fitness function value per iteration, potentially hindering the PSO convergence.

We develop a rsPSO algorithm that is able to substantially reduce the total fitness evaluation complexity while reaping satisfactory AASR performance. Specifically, we introduce a recursive sampling surrogate function to replace the fitness function in \eqref{eq:60}.
We partition all $B$ random channel samples into $N_B$ batches with $B_s = B/N_B$ samples per batch.
At each iteration, the mini-batch sampling surrogate function is given by
\begin{align}
J^{(i)} (\boldsymbol{\theta}) = \bar{\mu}^{(i)} J^{(i-1)}   + \mu^{(i)} \sum_{m=1}^{M} \sum_{t=1}^{T} \sum_{B=(i-1)B_s+1}^{i B_s} \frac{\log_2 \left[ 1 +  \gamma_m \left( P_m[t], \boldsymbol{\theta}, \mathcal{H}^{b}[t-\tau] \right)  \right] }{T B_s}, \label{eq:63}
\end{align}
where $\bar{\mu}^{(i)}=1-\mu^{(i)}$, $\mu^{(i)}$ denotes a decay weight coefficient associated with the new sampling of the AASR.
To guarantee more accurate sampling approximation as $i$ increases in the stochastic optimization, $\mu^{(i)}=i^{-0.2}$ is commonly used according to the diminishing stepsize rules \cite{7412752}.

The proposed rsPSO algorithm is summarized in Algorithm \ref{alg:3}. Given its use of random channel samples, the algorithm is \emph{data-driven} in spirit.
Given the S-CSI estimated at the beginning of a frame, the BS can generate the $L_B$ channel samples to approximate the future channels in the frame used for downlink transmissions. Recall that there are $T$ slots in the frame. The BS can generate the NLOS components of $\frac{B}{T}$ channel samples, and therefore, we set $T=I_{\mathrm{iter}}$ without loss of generality.
By repeating this process from $t=1$ to $T$, we can obtain $B$ time-varying channel samples to approximate CSI measurements of the upcoming $T$ slots.

\begin{algorithm}[h!]
\caption{The proposed rsPSO Algorithm}
\label{alg:3}
  \SetKwBlock{Begin}{Step 1:}{end}
  \Begin(\textbf{Swarm initialization}){
  \SetKwInOut{KwIn}{\textbf{Initialize the PSO control parameters}} \KwIn{$\{P, w, c_1, c_2, \varepsilon_1, \varepsilon_2, I_{\mathrm{iter}}\}$.}
  \SetKwInOut{KwIn}{\textbf{Initialize the position and velocity}} \KwIn{$\{\boldsymbol{\theta}_p^{(0)}\}_{p=1}^{P},  \{\mathbf{v}_{p}^{(0)}\}_{p=1}^{P}$.}
  Generate $B$ outdated channel samples according to the S-CSI and first-order AR model\;
  Separate the total channel samples into $N_B$ mini-batches with each $B_s$ samples\;
  For each particle $\boldsymbol{\theta}_p^{(0)}$, obtain the associated optimal powers $\mathbf{p}^{(0)}_p$ by using \eqref{eq:57}\;
  Evaluate the fitness value in \eqref{eq:63} of each particle using the first mini-batch of samples, and obtain the initial best position set $\{\hat{\boldsymbol{\theta}}_p\}_{p=1}^{P}$\;
  Find the global optimal position $\hat{\boldsymbol{\theta}} = \arg\max_{\hat{\boldsymbol{\theta}}} \{ J^{(0)} (\hat{\boldsymbol{\theta}}_1), J^{(0)} (\hat{\boldsymbol{\theta}}_2), \cdots, J^{(0)} (\hat{\boldsymbol{\theta}}_P) \}$\;}
  \SetKwBlock{Begin}{Step 2:}{end}
  \Begin(\textbf{Iterative search}){
  \For{$i=1: I_{\mathrm{iter}}$}{{
  Update particle velocity $\{\mathbf{v}_{p}^{(i)}\}_{p=1}^{P}$ and position $\{\boldsymbol{\theta}_p^{(i)}\}_{p=1}^{P}$ by using \eqref{eq:61} and \eqref{eq:62}\;
  \For{$p=1: P$}{
  Given particle $\boldsymbol{\theta}_p^{(i)}$, calculate the optimal powers $\mathbf{p}^{(i)}_p$ according to \eqref{eq:57}\;
  Evaluate the fitness value of particle $\boldsymbol{\theta}_p^{(i)}$ over $i^{\mathrm{th}}$ mini-batch of samples by using (\ref{eq:63})\;
  \eIf{$J^{(i)} > J^{(i-1)}$}
  {$\hat{\boldsymbol{\theta}}_p = \boldsymbol{\theta}_p^{(i)}$\;}{$\hat{\boldsymbol{\theta}}_p = \boldsymbol{\theta}_p^{(i-1)}$\;}
  Update $\hat{\boldsymbol{\theta}}_p = \arg\max_{\hat{\boldsymbol{\theta}}} \{ J^{(i)} (\hat{\boldsymbol{\theta}}_1), J^{(i)} (\hat{\boldsymbol{\theta}}_2), \cdots, J^{(i)} (\hat{\boldsymbol{\theta}}_P) \}$\;
  }
  }
  }
  }
\end{algorithm}

Finally, we analyze the computational complexity of the proposed rsPSO algorithm measured by floating point operations per second (FLOPS).
According to \eqref{eq:55}, the operations dominating the computational complexity of each particle include a matrix addition, four matrix multiplications, and a matrix inversion per channel sample per iteration, requiring $F_1 = 4 \left(N_{\mathrm{t}}+K \right)N^2 +  4M N_{\mathrm{t}} K + 4M^2  K + (4M^3+M^2+M)$ FLOPS \cite{4917823}.
The overall computational complexity of the algorithm is $P \times B_s \times F_1$ FLOPS, according to \eqref{eq:63}. In contrast, a direct use of all $B$ channel samples to evaluate the fitness function per iteration using \eqref{eq:59} incurs $P \times B_s \times F_1$ FLOPS, which is $N_B$ times the complexity of the new rsPSO algorithm since $B=N_B B_s$.

\section{Results and Discussions}\label{sec:7}
To demonstrate the potential of using IRSs for MIMO networks, in this section, we provide numerical results to verify the effectiveness and benefits of the above-proposed transmission architecture and resource allocation.
In particular, the BS is located in the origin point, and MUs are uniformly and randomly distributed in a circle centered at (40 m, 0 m) with a radius of 10 m. The total transmit power is $P_{\max}=30$ dBm.

\begin{figure}[h]
	\centering{}\includegraphics[width=4in]{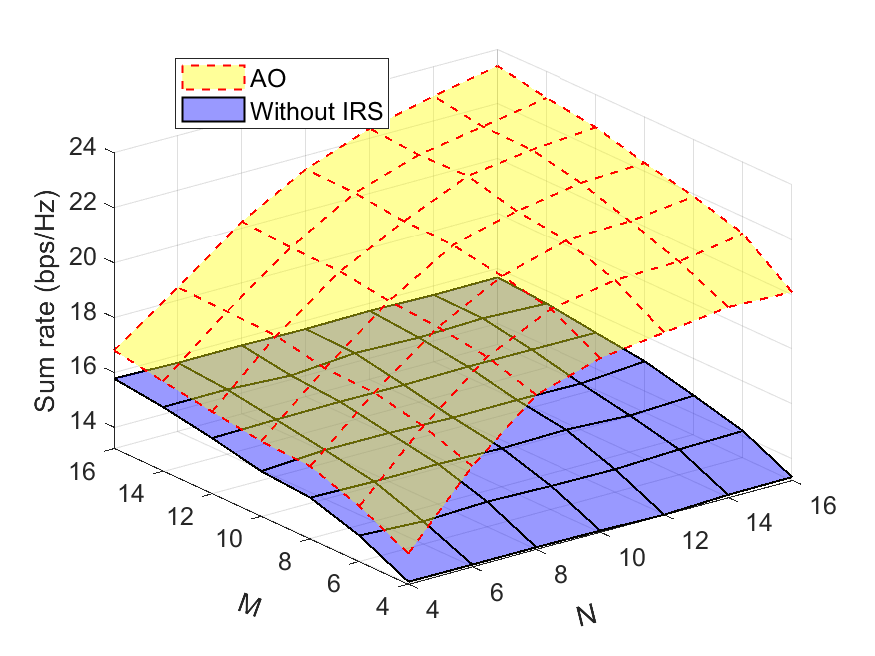}
	\caption{Achievable sum-rate versus the numbers of BS antennas $M$ and IRS elements $N$.}
	\label{fig:10}
\end{figure}

Firstly, we investigate the passive beamforming gains of IRS in the presence of direct BS-MU links.
Fig. \ref{fig:10} shows the impact of the number of BS antennas $M$ and IRS elements $N$ on the achievable sum-rate. The sum-rate of the proposed algorithm with the IRS exhibits superiority over the scheme without the IRS. The achievable sum-rate grows as $M$ and $N$ increase. We see that the achievable sum rate with $N$ grows faster than that with $M$. More importantly, this implies that the higher rate gain can be achieved by low-cost passive beamforming of properly deployed IRSs instead of active but expensive massive MIMO antennas.
Therefore, the IRS can be viewed as a promising candidate for massive MIMO 2.0 technology.

\begin{figure}[h]
	\centering{}\includegraphics[width=4in]{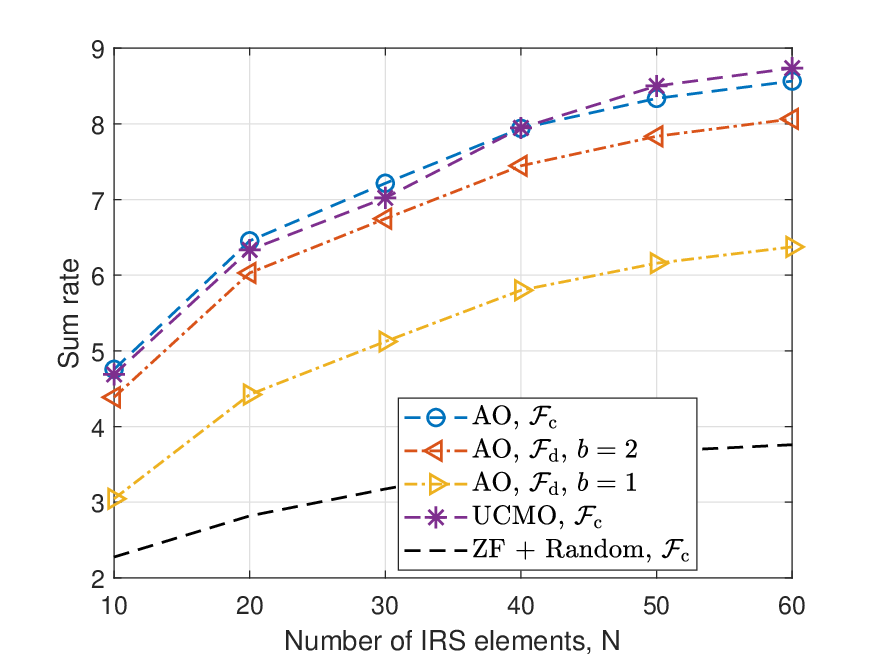}
	\caption{Achievable sum-rate versus the IRS elements $N$ of different methods.}
	\label{fig:11}
\end{figure}

Considering that the reflective elements cannot be implemented with the infinite resolution values in practice, the set of the low-resolution phase shifts at each element is given by
\begin{align}
\mathcal F_{\text{d}} = \left\{ \theta_{g,m} = e^{\mathrm{j} \varphi_{g,m}} \Big\vert
\varphi_{g,m} \in \left\{  \frac{2\pi i}{2^B}  \right\}_{i=0}^{2^B-1} \right\}, \label{eq:64}
\end{align}
where $B$ denotes the phase resolution in number of bits, and $B \rightarrow \infty$ corresponds to the continuous phase-shift set $\mathcal F_{\text{c}}$.
In Fig. \ref{fig:11}, we plot the sum-rate versus the number of reflecting elements under the single-IRS system. ZF precoding with a random reflection pattern method is selected as the baseline. We see that the proposed AO and UCMO algorithms outperform the baseline method in terms of sum-rate.
Besides the performance, the complexity superiority of the proposed UCMO approaches over the SDR method.
In addition, the 1-bit phase-shifter still stably precedes the baseline method, while the 2-bit phase-shifter has an overwhelming advantage.

\begin{figure}[h]
	\centering{}\includegraphics[width=4in]{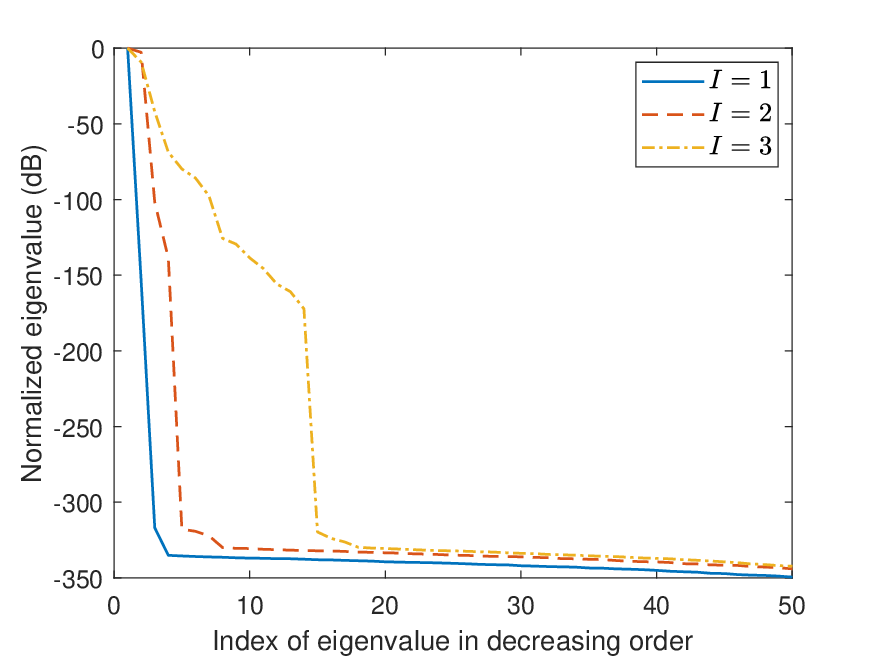}
	\caption{Normalized eigenvalues of the channel correlation matrix versus the number of IRS units.}
	\label{fig:12}
\end{figure}

In Fig. \ref{fig:12}, we assess the impact of distributed IRS deployment on the channel rank. To characterize the channel DoF, we adopt the normalized eigenvalues of the channel correlation $\tilde{\mathbf{H}} \tilde{\mathbf{H}}^{\sf H}$ as the indicator. For single-IRS cases where only IRS unit 1 is activated, only 2\textasciitilde 3 valid independent channels can carry signal streams. This is due to the sparse BS-IRS channel, which can serve up to 3 users. The number of valid independent channels is expanded for distributed-IRS cases with $I > 1$. This confirms the channel rank improvement of the distributed IRS scheme.

\begin{figure}[h]
	\centering{}\includegraphics[width=4in]{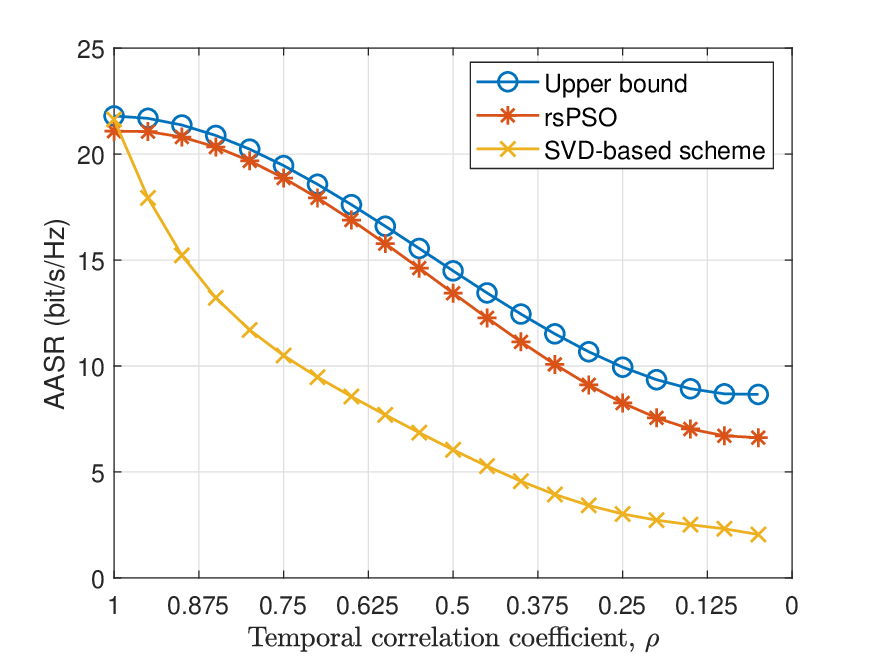}
	\caption{AASR against temporal correlation coefficient $\rho$.}
	\label{fig:13}
\end{figure}

Finally, we validate the two-timescale reflection optimization method by considering time-varying MIMO channels.
Fig. \ref{fig:13} presents the AASRs achieved by the different schemes against the temporal correlation coefficient $\rho$.
A larger temporal correlation coefficient $\rho$ indicates a smaller time delay $\tau$, thus resulting in a smaller channel error between consecutive slots.
In the extreme case where the O-CSI is independent of the I-CSI, the precoding matrices of the schemes are independent of the I-CSI and equivalent to random precoding schemes.
We see that more severe O-CSI leads to a lower AASR. This confirms the importance of exploiting the channel correlation between consecutive time slots to improve the AASR performance.
Under the nearly independent channel conditions, i.e., $\rho \rightarrow 0$, the SVD precoder based on the O-CSI is, in essence, an arbitrary precoder. Nevertheless, the proposed SVD-ZF-based scheme can still outperform the SVD-based scheme based on the O-CSI because of its use of the ZF detector to suppress inter-stream interference based on the I-CSI at the UE.
In the case of $\rho=1$, the CSI remains unchanged and never outdated throughout all time slots. It is seen that the conventional SVD scheme can achieve the same AASR as the upper bound since SVD precoding is known to be optimal under the I-CSI.

\section{Conclusion}\label{sec:8}
In this chapter, slot-by-slot and two-timescale reflection pattern design schemes are designed. Firstly, the distributed IRS-assisted mmWave massive MIMO system is built, which can overcome the channel blockage with passive meta-atoms and improve channel rank to enhance the signal coverage and channel DoF.
Then, slot-by-slot IRS reflection pattern optimization is formulated to reap the maximum system weighted sum-rate in the considered mmWave system. To maximize the sum-rate, the precoder at the BS and the reflection pattern of the distributed IRS are jointly optimized in each time slot. The AO framework is developed, where the popular SDP algorithm and the low-complexity UCMO algorithm are proposed to obtain the reflection pattern, respectively.
Next, to avoid the heavy overhead configuration signaling and frequently cascaded channel estimation in time-variant MIMO channels, a two-timescale reflection pattern optimization strategy is devised to maximize the long-term AASR of the IRS-assisted MIMO system. Specifically, the small-timescale transmit beamformer and power allocation at the BS are updated in every slot with the SVD-ZF scheme to mitigate the performance deterioration resulting from outdated I-CSI. An rsPSO algorithm is proposed to solve the large-timescale IRS reflection pattern efficiently. The proposed rsPSO algorithm can dramatically reduce the total fitness evaluation complexity with light sampling of random channels.

\section{Future Work}\label{sec:9}
Recent research in the realm of IRS-assisted wireless networks has indicated the important role of the IRS in the convergence of communication, sensing and computing (CSC) for future mobile systems.
However, The study domain of the IRS-assisted CSC system is still in its infancy.
This section highlights possible research directions for the future investigation of the IRS. 

\begin{itemize}
    \item For the multi-cell multi-IRS scenario, each cellular network is deployed by a distinct network operator. BSs or cell-edge MUs may be interfered with by some IRSs from other operators, as IRSs cannot realize strictly narrow-band response. Thus, interference management between different operators is a vital research direction for IRS-assisted communications.
    \item Some studies on IRS-enabled modulation and precoding schemes have initially shown the potential of wave-based computing~\cite{8928065, 10158690}. Holographic IRS can generate desired beams or signals by controlling the radiation amplitude and phase of the EM wave. The resulting computing ability may not meet the complicated tasks for the current low-cost phase-shift-only IRS. For this reason, active IRS or few-layer IRS can provide more wave-control DoFs.
    \item Since IRSs can provide virtual LoS links, facilitating the sensing functions for the integrated sensing and communication (ISAC) network. Moreover, the near-field channel may be necessary for the extremely large IRS, and a near-field LoS channel can be rank-sufficient due to the spherical wavefront. As near-field channel involves both distance and angle information, they can benefit both sensing and computing more. This implies that the related information-theoretic limits analysis and near-field range extension are worth noting in future ISAC directions.
    \end{itemize}






\ifx\weAreInMain\undefined                                               
\bibliographystyle{IEEEtran}
\bibliography{bibchapter.bib}                                
\end{document}                                                           
\fi                                                                      